\documentclass[5p, twocolumn]{elsarticle}
\usepackage{ amssymb }
\usepackage[T1]{fontenc}
\usepackage[english]{babel}
\usepackage[cmintegrals]{newtxmath}
\usepackage[linesnumbered,ruled,vlined]{algorithm2e}
\usepackage{algorithmic}
\usepackage{graphicx}
\usepackage{textcomp}
\usepackage{xcolor}
\usepackage{siunitx} 
\usepackage{xspace}
\usepackage{tikz}
\usepackage{pgf}
\usepackage{anyfontsize}
\usepackage{ wasysym }
\usepackage{cancel}
\usepackage[utf8]{inputenc}\DeclareUnicodeCharacter{2212}{-}
\usepackage{enumitem} 
\usepackage{hyperref}
\usepackage{array}
\usepackage{booktabs} 
\usepackage{bigstrut}
\usepackage{floatrow}
\usepackage{xcolor,pifont}
\newtheorem{definition}{Definition}
\newtheorem{lemma}{Lemma}
\newenvironment{proof}{\noindent \textit{Proof:}}{\hfill$\square$}

\newcommand{\red}[1]{\textcolor{red}{#1}}

\definecolor{darkgreen}{rgb}{0.0, 0.5, 0.0}
\newcommand*\colourcheck[1]{%
  \expandafter\newcommand\csname #1check\endcsname{\textcolor{#1}{$\checked$}}%
}
\colourcheck{blue}
\colourcheck{green}
\colourcheck{red}
\colourcheck{darkgreen}

\definecolor{darkred}{rgb}{0.9, 0.17, 0.31}
\newcommand*\colourcross[1]{%
  \expandafter\newcommand\csname #1cross\endcsname{\textcolor{#1}{$\times$}}%
}
\colourcross{blue}
\colourcross{green}
\colourcross{red}
\colourcross{darkred}

\newcommand*\colourtilde[1]{%
  \expandafter\newcommand\csname #1tilde\endcsname{\textcolor{#1}{$\thicksim$}}%
}
\colourtilde{blue}
\colourtilde{green}
\colourtilde{red}
\colourtilde{orange}

\makeatletter
\newlength\mylena
\newlength\mylenb
\newcommand\mystrut[1][2]{%
    \setlength\mylena{#1\ht\@arstrutbox}%
    \setlength\mylenb{#1\dp\@arstrutbox}%
    \rule[\mylenb]{0pt}{\mylena}}
\makeatother
\usepackage{pifont}
\usepackage[lofdepth,lotdepth]{subfig}

\definecolor{darkorchid}{rgb}{0.6, 0.2, 0.8}\xspaceaddexceptions{[]}

\SetCommentSty{mycommfont}
\graphicspath{{figures/}}
\SetArgSty{textnormal}
\SetKw{extend}{extend\_and\_filter}
\SetKw{orb}{or}
\SetKw{continue}{continue}
\SetKw{addroute}{add\_route}
\SetKw{removedominated}{remove\_dominated}
\SetKw{linearclean}{linear\_clean}
\SetKw{hoppingclean}{hopping\_clean}
\SetKw{int}{int}
\SetKw{add}{add}
\SetKw{to}{to}
\SetKw{sort}{sort}
\SetKw{short}{short}
\SetKw{mergesortd}{mergesortd1}
\SetKw{byte}{byte}
\SetKw{kwand}{and}
\SetKw{true}{True}
\SetKw{false}{False}
\SetKw{not}{not}
\SetKw{flush}{flush}
\SetKw{first}{first}
\SetKw{last}{last}
\SetKw{extend}{ExtendPaths}
\SetKw{compute}{CptExtendablePaths}
\SetKw{succ}{succ}
\SetKw{kwmax}{max}

\newcommand\var[1]{\texttt{#1}}
\newcommand\extendableIndex{\var{extendable}\xspace}
\newcommand\extendableIndexbis{\var{nextextendable}\xspace}
\newcommand\extendableIndexbisv{\var{nextextendable\_v}\xspace}
\newcommand\vardist{\var{dist}\xspace}
\newcommand\vardistv{\var{dist\_v}\xspace}
\newcommand\varpfront{\var{pfront}\xspace}
\newcommand\varpfrontiv{\var{pfront\_iv}\xspace}
\newcommand\varpfcand{\var{pf\_cand}\xspace}

\newcommand\varn{\var{n}\xspace}
\newcommand\varu{\var{u}\xspace}
\newcommand\varl{\var{l}\xspace}
\newcommand\vard{\var{d2}\xspace}
\newcommand\vardd{\var{d1}\xspace}
\newcommand\vardu{\var{d2u}\xspace}
\newcommand\varddu{\var{d1u}\xspace}
\newcommand\vardv{\var{d2v}\xspace}
\newcommand\varddv{\var{d1v}\xspace}
\renewcommand\varv{\var{v}\xspace}
\newcommand\vari{\var{i}\xspace}
\newcommand\varimax{\var{imax}\xspace}
\newcommand\varnb{\var{nb}\xspace}
\newcommand\vardlist{\var{d\_list}\xspace}
\newcommand\varlastd{\var{last\_d2}\xspace}
\newcommand\varmaxd{\var{max\_d1}\xspace}

\newcommand{\ie}{\emph{i.e.,~}}
\newcommand{\eg}{\emph{e.g.,~}}
\let\oldnl\nl
\newcommand{\nonl}{\renewcommand{\nl}{\let\nl\oldnl}}
\newcommand\bestcop{\var{BEST2COP}\xspace}
\newcommand\bestcope{\var{BEST2COP-E}\xspace}
\newcommand\samcra{\var{SAMCRA}\xspace}
\newcommand\samcrasrg{\var{SAMCRA+SRG}\xspace}
\newcommand\samcralca{\var{SAMCRA+LCA}\xspace}
\renewcommand\red[1]{#1}

\ExplSyntaxOn

\NewDocumentCommand{\makepath}{ m }
 {
 \ensuremath{
      \clist_set:Nn \mylist { #1 }
      \clist_pop:NN \mylist \listtmpvar
      \int_do_until:nNnn {\clist_count:N \mylist} = {0} {

          \listtmpvar{-}
          \clist_pop:NN \mylist \listtmpvar
      }
      \listtmpvar
  }
 }

\NewDocumentCommand{\makesrpath}{ m }
 {
 \ensuremath{
      \clist_set:Nn \mylist { #1 }
      \clist_pop:NN \mylist \listtmpvar
      \int_do_until:nNnn {\clist_count:N \mylist} = {0} {

          \listtmpvar{|}
          \clist_pop:NN \mylist \listtmpvar
      }
      \listtmpvar
  }
 }

\ExplSyntaxOff

\begin{document}

\title{Deploying Near-Optimal Delay-Constrained Paths with Segment Routing in Massive-Scale Networks}

\affiliation[1]{organization={ICube, University of Strasbourg},
country={France}}
\affiliation[2]{organization={Cisco Systems},
country={California}}

\author[1]{Jean-Romain Luttringer\corref{cor1}}
\ead{jr.luttringer@unistra.fr}

\author[1]{Thomas Alfroy}
\ead{talfroy@unistra.fr}

\author[1]{Pascal Mérindol}
\ead{merindol@unistra.fr}

\author[1]{Quentin Bramas}
\ead{bramas@unistra.fr}

\author[2]{François Clad}
\ead{fclad@cisco.com}

\author[1]{Cristel Pelsser}
\ead{pelsser@unistra.fr}

\cortext[cor1]{Corresponding author}

\begin{abstract}
  With a growing demand for quasi-instantaneous communication services such
  as real-time video streaming, cloud gaming, and industry 4.0
  applications, multi-constraint Traffic Engineering (TE) becomes increasingly
  important. While legacy TE management planes like MPLS have proven
  laborious to deploy, Segment Routing (SR) drastically
  eases the deployment of TE paths and is thus increasingly adopted by Internet Service Providers (ISP).
  There is now a clear need in computing and deploying Delay-Constrained Least-Cost
  paths (DCLC) with SR for real-time interactive services requiring both low delay and high bandwidth routes.
  However, most current DCLC solutions are not tailored for SR.
  They also often lack efficiency (particularly exact schemes) or guarantees (by relying on unbounded heuristics). Similarly to approximation schemes, we argue that the
  actual challenge is to design an algorithm providing both performances and strong guarantees. However, conversely to most of these schemes, we also consider operational constraints to provide a practical, high-performance implementation.

  In this work, \red{we extend and further evaluate our previous contribution, \bestcop. \bestcop leverages inherent limitations in the accuracy of delay measurements, accounts for the operational constraint added by SR, and provides guarantees and bounded computation time in all cases thanks to simple but efficient data structures and amortized procedures. 
  We show that \bestcop is faster than a state-of-the-art algorithm on both random \emph{and} real networks of up to \num{1000} nodes. Relying on commodity hardware with a single thread, our algorithm retrieves all
  non-superfluous 3-dimensional routes in under 100ms in both cases.
  This execution time is further reduced using multiple threads, as the design of \bestcop enables a speedup almost linear in the number of cores thanks to a balanced computing load.}
  Finally, we extend \bestcop to deal with massive-scale ISP by leveraging the multi-area partitioning of these deployments.
  Thanks to our new topology generator specifically designed to model realistic patterns in such massive IP networks, we show that \bestcop can solve DCLC-SR in approximately 1 second even for ISP having more than \num{100000} routers.
\end{abstract}


\maketitle

\section{Introduction}\label{sec:intro}


Latency is critical in modern networks for various applications. The constraints on the delay are indeed increasingly stringent.
For example, in financial networks, vast amounts
of money depend on the ability to receive information in real-time.
Likewise, technologies such as 5G slicing, in addition to requiring significant bandwidth availability,
demand strong end-to-end delay guarantees depending on
the service they aim to provide, e.g., less than 15ms for low latency applications such as motion control for industry 4.0, VR, or video games~\cite{5GSlicing}.
For such interactive applications, the latency is as critical as the \textit{IGP cost}.

This Interior Gateway Protocol cost is defined as an additive metric that usually reflects both the
link's bandwidth and the operator's load
distribution choices on the topology.
Paths within an IGP are computed by minimizing this cost.
Thus, although delay constraints are increasingly important, they should not be enforced
to the detriment of the IGP cost.
With minimal IGP distances, the traffic benefits from high-bandwidth links and follows the
operator's intent in managing the network and its load.
With bounded delays, the traffic can benefit from paths allowing for sufficient
interactivity.
It is thus relevant to minimize the IGP cost while enforcing an upper constraint on the latency.
Computing such paths requires to solve DCLC, an NP-Hard problem standing for Delay Constrained Least Cost.\\

\begin{figure}[!ht]
    \centering
    \includegraphics[scale=0.3]{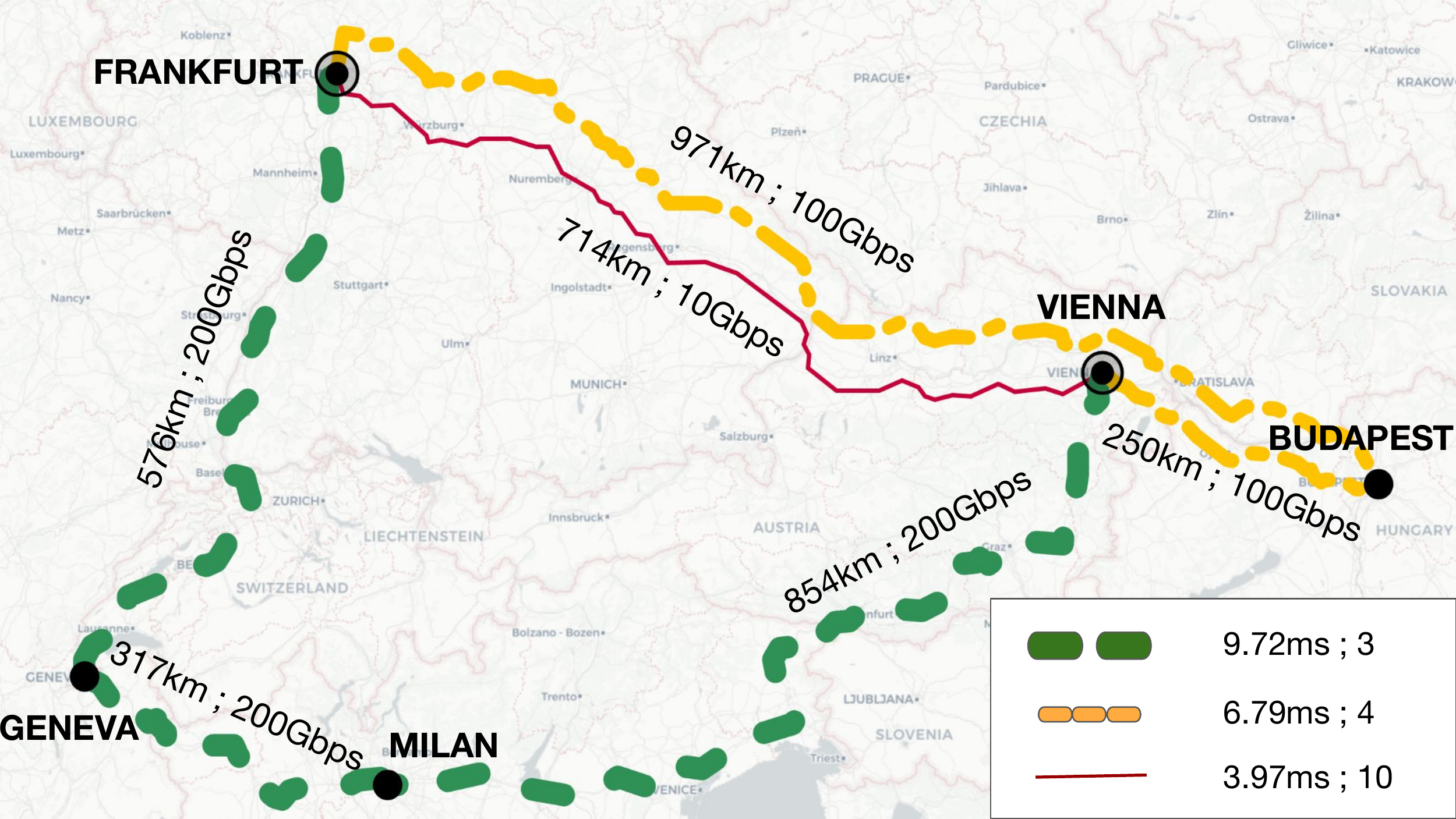}
    \caption{Practical relevance of DCLC in the GEANT network. IGP costs are deduced from the bandwidth of each link. Depending on their needs (in terms of delay and bandwidth), applications can opt for three non-comparable paths between Frankfurt and Vienna.}
    \label{fig:DCLCGeant}
\end{figure}

\textbf{DCLC: a relevant issue.}
\red{Although one may expect the two metrics to be strongly correlated in practice, there are various cases where the delay and the IGP cost may be drastically different. For example, the IGP cost may have been tuned arbitrarily by the operator. Heterogeneous infrastructures between countries or geographical constraints may also create this effect. This can be illustrated on real networks, as displayed by
Fig.~\ref{fig:DCLCGeant}.}
This map is a sample of the GEANT transit network~\cite{GeantTopo}.
As fibers often follow major roads, we rely on real road distances to infer the propagation delay of each link while the bandwidth, and so the estimated IGP cost, matches the indications provided by GEANT. 
A green link has an IGP cost of 1 while the IGP cost is 2 and 10 respectively for the yellow and pink ones.


\red{Note that the two metrics are not correlated, hence all three paths shown between Frankfurt and Vienna offer diverse interesting options. They are non-comparable (or \textit{non-dominated}) paths and form the \textit{Pareto-front} of the paths between the two cities.}
Either solely the delay matters and the direct link (in pink) should be preferred, or the ISP prefers to favor high capacity links, and the green path, minimizing the IGP cost, should be used.
The yellow path, however, offers an interesting compromise.
Out of all paths offering a latency well-below 10ms, it is the one minimizing the IGP cost. Thus,
it allows to provide strict Service-Level Agreement
($< 9$ms), while considering the IGP cost.
These kind of paths, retrieved by solving DCLC, provide more options by enabling tradeoffs between the two most important networking metrics. Applications such as videoconferences, for example, can then benefit both from real-time interactive voice exchange (delay) and high video quality (bandwidth).
In addition, IGP costs are also tuned to represent the operational costs. Any
deviation from the shortest IGP paths thus results in additional costs for the operator. For all these reasons, there exist a clear interest
for algorithms able to solve and deploy DCLC paths~\cite{SegmentR71:online}.
However and so far, while this problem has received a lot of attention in the last decades from the network research community~\cite{survey2018,survey2010}, no technologies were available for an efficient deployment of such paths.\\

\red{\textbf{Segment Routing, MSD \& DCLC's (large-scale) rebirth.} Segment Routing (SR) is a vibrant technology gathering traction from router vendors,
network operators and academic communities \cite{I-D.matsushima-spring-srv6-deployment-status,ventre_segment_2020}.
Relying on a combination of strict and loose source routing, SR enables to deviate the traffic from
the shortest IGP paths through a selected set of nodes and/or links by prepending routing instructions to the packet itself}.
Such deviations may for example allow to route traffic through a path with lower latency.
These deviations are encoded in the form of \textit{segments} within the packet itself. To prevent any packet forwarding degradation, the number of deviations \red{(\ie instructions)} one can encode is limited \red{to MSD (Maximum Segment Depth), whose exact value depends on the hardware}. While this technology is adequate to support
a variety of services, operators mainly deploy SR in the hopes of performing fine-grained and ECMP\footnote{Equal Cost Multi-Path}-friendly tactical Traffic-Engineering (TE)~\cite{SR4TE}, due to its reduced overhead compared to RSVP-TE \cite{filsfils2017segment}. Our discussion with network vendors further revealed a clear desire from operators to efficiently compute DCLC paths deployable with Segment Routing~\cite{SegmentR71:online}. Such a solution should thus not only
encompass Segment Routing, but also fare well on large-sized networks of several thousand of nodes, \red{as already observable in current SR deployments\cite{I-D.matsushima-spring-srv6-deployment-status}. Indeed, while we showed in~\cite{nca2020} that computing DCLC paths for Segment Routing (DCLC-SR) is possible in far less than a second on networks of up to $1000$ nodes, scaling to ten or a hundred times more routers remains an open issue.\\}


\textbf{Challenges.}
\red{In this paper, we extend our previous contribution, \bestcop (Best Exact Segment Track for the 2-Constrained Optimal Paths)~\cite{nca2020} to make DCLC-TE possible on (very) large scale networks.}
\red{In order to achieve our goal, we solve several challenges: (i) efficiently cap, or even minimize, the number of segments
to consider the packet manipulation overhead supported by routers, (ii) scale with very large modern networks, and (iii) provide near-exact algorithms
with bounded error margin and strong guarantees despite the difficulty induced by considering
three metrics (cost, delay, and number of segments)\footnote{While difficult instances are unlikely to occur in practice, our algorithm can tackle even such corner cases. Thus, our proposal is not only efficient in practice but provably correct and efficient even for worst theoretical cases.}. We now briefly explain the main design aspects allowing us to solve these challenges. The latter will be explained more thoroughly in their dedicated sections.} \\

\textbf{Efficiently encompassing the Maximum Segment Depth constraint (MSD).}
\red{An SR router can only prepend up to \textit{MSD} routing instructions to a packet at line-rate, i.e., $\approx 10$ with the best current hardware.}
\red{Although we find in our new study, shown in Section~\ref{ssec:nbseg} that this limit does not prevent from deploying most DCLC paths in practice (if not all in easy cases), this constraint must still be taken into account. If ignored, the computed paths have no guarantees to be deployable, as they may exceed MSD.}

\red{While this adds an additive metric to consider (the number of segments), \bestcop manages, through adequate data-structures and graph exploration, to natively manipulate the list of segments  
and ensure that paths requiring more than MSD segments are natively removed from the exploration space.}\\


\textbf{Towards massive-scale networks.}
\red{Our new proposal, \bestcope (\bestcop Extended), aims
at enabling fine-grained TE on massive-scale networks efficiently using a divide-and-conquer approach.}
Indeed, massive networks usually rely on a standard
physical and logical partitioning, as IGP protocols do not scale well as is.
By leveraging this decomposition and re-designing \bestcop to benefit from multi-threaded architectures, it becomes possible to solve DCLC-SR in a time suited for real-time routing.
To evaluate our contribution, we create a topology generator, YARGG, able to construct realistic massive-scale, multi-valuated, and multi-area topologies based on geographical data. In such cases, our extension solves DCLC-SR in $\approx$ 1 second for $\approx$ \num{100000} nodes.\\

\red{\textbf{Bounded error margin and practical concerns.}}
\red{DCLC is a well-known NP-Hard problem~\cite{536364}.
While there exist several ways to solve DCLC ~\cite{survey2018, survey2010}, they usually do not consider the underlying deployment technologies and real-life deployment constraints. We keep the latter at the core of our design. The nature of the concerned  TE paths allows us to consider a stable latency metric (the propagation delay). This is essential as
unstable metrics should not be considered nor advertised when performing routing~\cite{RFC7471,RFC8570}.
Furthermore, because of the arbitrary nature of the latency constraint and the inherent imprecision of
the delay measurement, we argue that an acceptable error margin regarding the delay constraint is acceptable
(more so than on the IGP cost). Our algorithm is designed to take advantage, if needed, of this acceptable
margin to returns the DCLC path efficiently in all situations to all destinations, with strong guarantees.\\}

\textbf{Summary and Contributions.}
By taking into consideration the operational deployment of constrained paths with SR, the current scale and structure of modern networks, as well as the practicality of the delay measures and constraints, \textbf{we designed BEST2COPE (BEST2COP Extended), a simple efficient algorithm able to solve DCLC-SR in $\approx$ 1s in large networks of $\approx$\num{100000} nodes}.

The main achievements of our proposal follow the organization of this paper:
\begin{itemize}
    \item In Section~\ref{sec:background}, we present SR in further details. \red{Compared to~\cite{nca2020}, we discuss the context and works related to DCLC in further detail. We also add an evaluation regarding the relevance of SR for deploying DCLC paths};
    \item In Section~\ref{sec:best2cop}, we formalize DCLC-SR and also its generalization (2COP). In particular, we describe the network characteristics we leverage and define the construct we use to encompass SR (in Sec.~\ref{sec:true} and Sec.~\ref{sec:sr-graph} respectively). \red{In Sec.~\ref{sec:flat}, we briefly summarize \bestcop (initially introduced in~\cite{nca2020} and detailed here in an appendix) and emphasize how it can easily benefit from multi-threaded architectures}.
\end{itemize}

\red{Besides the more detailed background and the evaluation of SR relevance, the main new contributions in this extension come in the last three sections:}

\begin{itemize}
    \item In Sec.~\ref{sec:areas}, we extend \bestcop to \bestcope, to deal with massive scale networks relying on area partitioning (as with OSPF with a single metric), making it able to solve DCLC efficiently in graphs of $\approx \num{100000}$ nodes;
    \item In Sec.~\ref{sec:complexity}, we formally define the guarantees brought by \bestcop and its versatility to solve several TE optimization problems at once (\ie the sub-problems of 2COP), as well as its polynomial complexity;
    \item Finally, in Section~\ref{sec:evaluation}, we
    present our large-scale topology generator, YARGG, and evaluate \bestcope on the resulting topologies and compare our proposal to the relevant state-of-the-art path computation algorithm.
\end{itemize}

\section{Back to the Future: DCLC vs SR}\label{sec:background}

\subsection{Segment Routing Background and Practical Usages}

\red{Segment Routing implements source routing by prepending packets with a stack of up to MSD \textit{segments}. In a nutshell, segments are checkpoints the packet has to go through.
There are two main types of segments}:
\begin{itemize}
    \item \textbf{Node segments.} A node segment $v$ indicates that the packet should (first) be forwarded to $v$ with ECMP (instead of its final IP destination). Flows are then load-balanced among the best IGP next hops for destination $v$.
    \item \textbf{Adjacency segments.} Adjacency segments indicate that the packet should be forwarded through a specific interface and its link.
\end{itemize}

\red{Once computed, the stack of segments encoding the desired path is added to the packet. Routers forward packets according to the topmost segment, which is removed from the stack when the packet reaches the associated intermediate destination.}


Adjacency segments may be globally advertised, and thus be used the same way as node segments, or they may only have a local scope and, as such, can only be interpreted by the router possessing said interface. In this case, the packet should first be guided to the corresponding router, by prepending the associated node segment. \red{In the following, as a worst operational case, we consider the latter scenario, as it
always requires the highest number of segments.}


Segment Routing attracted a lot of interest from the research community. A table
referencing most SR-related work can be found in Ventre et al.~\cite{ventre_segment_2020}. While some SR-TE works are related to tactical TE problems (like minimizing the maximum link utilization) taking indirectly into account some delay concerns~\cite{swarm, 10.1007/978-3-319-23219-5_41}, most of the works related to SR do not focus on DCLC, but rather bandwidth optimization ~\cite{7218434,Gay2017ExpectTU, Loops}, network resiliency \cite{8406885, 7524551}, monitoring~\cite{SRv6PM, 7524410}, limiting energy consumption ~\cite{7057272} or path encoding (the translation of path to segment lists)~\cite{guedrez2016label,7417097}. Aubry~\cite{AubryPhd}
proposes a way to compute paths requiring less than MSD segments while optimizing
an additive metric in polynomial time. The number of segments required is then evaluated. This work,
however, considers only a single metric in addition to the operational constraints.
The problem we tackle
(\ie DCLC paths for Segment Routing) deals with two metrics (in addition to the operational constraints). This additional dimension drastically changes the problem,
which then becomes NP-Hard.
Some works use a construct similar to ours (presented in details in Sec.~\ref{sec:sr-graph}) in order to prevent the need to perform conversions from network paths to segment lists, \cite{Lazzeri2015EfficientLE} in particular. However, the authors of \cite{Lazzeri2015EfficientLE} do not pretend to solve DCLC and, as such, do not tune the structure the same way
(\ie they do not remove dominated segments, as explained later on), and simply use their construct to sort paths lexicographically.

As aforementioned, while operators seem to mainly deploy SR to perform fine-grained TE, to the best of our knowledge, no DCLC variant exists for specifically tackling SR characteristics and constraints (except for our contribution).
Using segments to steer particular flows allows however to deviate some TE traffic from the best IGP paths
in order to achieve, for example, a lower latency (and by extension solve DCLC). A realistic example is shown on Fig.~\ref{fig:DCLCGeant} where the node segment \textit{Vienna}, as well as considering Vienna as the destination itself, would result in the packets following the best IGP path from Frankfurt to Vienna, i.e., the green dashed path.
To use the direct link instead (in plain pink) and so minimize the delay between the two nodes of this example, the associated adjacency segment would have to be used as it enforces a single link path having a smaller delay than the best IGP one (including here two intermediary routers).
Finally, the yellow path, offering a non dominated compromise between both metrics (and being the best option if considering a delay constraint of 8ms), requires the use of the node segment \textit{Budapest} to force the traffic to deviate from its best IGP path in green.
Before converting the paths to segment lists (and actually deploy them with SR), such non-dominated paths need first to be explored.
Computing these paths while ensuring that the number of segments necessary to encode them remains under MSD is at least as difficult as solving the standard DCLC problem since an additional constraint now applies.




\subsection{DCLC (Delay-Constrained, Least-Cost), a Well-known Difficult Problem having many Solutions?}

DCLC belongs to the set of NP-Hard problems (as well as most related multi-constrained path problems).
Intuitively, solely extending the least-cost path is not sufficient, as the latter may exceed the delay constraint.
Thus, paths with greater cost but lower delays must be memorized and extended as well.
These \textit{non-dominated paths} form the \textit{Pareto front} of the solution, whose size may grow exponentially with respect to the size of the graph.
However, DCLC in particular, and related variants and extensions in general, does possess several interesting applications such as mapping specific flows to their appropriate paths (in terms of interactive quality).
Thus, these problems have been extensively studied in the past decades.
Many solutions have been proposed so far, as summarized in these surveys~\cite{survey02, survey2010, survey2018}: they range from to heuristics and approximations to exact algorithms, or even genetic approaches.\\

\textbf{Heuristics.} Because DCLC is NP-Hard, several polynomial-time heuristics have been designed to limit the worst-case computing time, but at the detriment of any guarantees.
For example, \cite{Fallback} only returns the least-cost or least-delay path
if one is feasible (\ie respect the constraints).
More advanced proposals try to explore the delay and cost space simultaneously, by either combining in a distributed manner the least-cost and least-delay subpaths~\cite{DCUR, DCR, SFDCLC} or by aggregating both metrics into one in a
more or less intricate manner.

Aggregating metrics in a linear fashion~\cite{Jaffe, Aneja_Nair_1978}
preserves the \textit{subpaths optimality principle} (isotonicity of best single-metric paths) and therefore allows to use standard shortest paths
algorithms. However, it leads to a loss of relevant information regarding the quality
and feasibility of the computed paths~\cite{536364}, in particular if the hull of the Pareto Front is not convex. Some methods try to mitigate this effect by using a $k$-shortest path approach to possibly find more feasible paths~\cite{DCBFKLARAC, larac}, but such an extension may result in a large increase of execution time and may not provide more guarantees.
Other heuristics rely on non-linear metric aggregation.
While it seems to prevent loss of relevant information, at first glance, such algorithms expose themselves to maintain all non-dominated paths (towards all nodes) as the isotonicity does not hold anymore (while it holds with linear metrics). Since the Pareto Front may be exponential with respect to the size of the graph, those algorithms either simply impose a hard limit on the number of paths that can be maintained (\eg TAMCRA~\cite{TAMCRA} and LPH~\cite{EBFA01}), or specifically chose the ones to maintain through previously acquired knowledge (HMCOP~\cite{Korkmaz_Krunz_2001}).
Finally, other works like~\cite{Feng_Makki_Pissinou_Douligeris_2002, Guo_Matta_1999} rely on heuristics designed to solve a variant of DCLC, the MCP problem (Multi-Constrained Paths, the underlying NP-Complete decision version of DCLC -- with no optimization objective). It mainly consists of sequential MCP runs using a conservative cost constraint iteratively refined.

Relying on heuristics is tempting, but their lack of guarantees can prevent to enforce strict SLAs even when a suitable path actually exists. 
One can argue it is particularly unfortunate, as DCLC is only \textit{weakly} NP-hard: it can be solved exactly in pseudo-polynomial time, \ie polynomial in the numerical value of the input~\cite{gareycomputers}.
Said otherwise, DCLC is polynomial in the smallest largest weight of the two metrics once translated to integers.
Consequently, it is possible to design FPTAS\footnote{Fully Polynomial Time Approximation Scheme.} solving DCLC while offering strong guarantees~\cite{Papadimitriou_Yannakakis_2000}.\\

\textbf{Approximations.}
\red{Numerous FPTAS have been proposed to solve DCLC and related constrained shortest path problems. All the following algorithms fall into this category.}
The common principle behind these schemes is to reduce the precision (and/or magnitude) of the considered metrics.
This can be performed either directly, by \textit{scaling and rounding} the weights of each link, or indirectly,
by dividing the solution space into intervals and only maintaining paths belonging to different intervals (\textit{Interval partitioning})~\cite{Sahni_1977}.
Scaling methods usually consider either a high-level dynamic programming scheme or a low-level practical Dijkstra/Bellman-Ford core with pseudo-polynomial complexity, and round the link costs to turn their algorithms into an FPTAS (see for example Hassin~\cite{Hassin_1992}, Ergun et al.~\cite{Ergun_Sinha_Zhang_2002} or Lorenz and Raz~\cite{Lorenz_Raz_2001} methods).
Goel et al.~\cite{goel}, in particular, chose to round the delay instead of the cost and can consider multiple destinations (as our own algorithm).

Most interval partitioning solutions explore the graph through a Bellman-Ford approach.
The costs of the paths are mapped to intervals, and only the path with
the lowest delay within a given interval is kept. The size of the intervals
thus introduces a bounded error factor~\cite{Hassin_1992, Tsaggouris_Zaroliagis_2009}.
In particular, HIPH~\cite{HIPH} offers a dynamic approach between an approximation and exact scheme. It proposes to maintain up to $x$ non-dominated paths for each node and stores eventual additional paths using an interval partitioning strategy. This allows the algorithm to be exact on simple instances (resulting in a limited Pareto front, \ie polynomial in the number of nodes, in particular when it is bounded by $x$) and offer strong guarantees on more complex ones. This versatility is an interesting feature, as most real-life cases are expected to be simple instances with a bounded Pareto front size, in particular because one of the metrics may be coarse by nature. 
For these reasons, not only approximation schemes can offer practical solutions (with a bounded margin error) but also exact algorithms (with controlled performance), as they may be viable in terms of computing time for simple real-life IP networking instances.\\


\textbf{Exact methods.} Numerous exact methods have indeed also been studied extensively to solve DCLC.
Some methods simply use a $k$-shortest path approach to list all paths within the Pareto front~\cite{Namorado_Climaco_Queiros_Vieira_Martins_1982, Paixao_Santos_2008}.
On the other hand, Constrained Bellman-Ford~\cite{CBF94} (ironically, also called Constrained Dijkstra as it uses a priority queue -- denoted PQ in the following) explores paths by increasing delays and lists all non-dominated paths towards each node.
Several algorithms use the same principle but order the paths differently within the queue, relying either on a lexicographical ordering, ordered aggregated sums, or a simple FIFO/LIFO ordering~\cite{Martins_1984, Martins_Santos_1999,Brumbaugh-Smith_Shier_1989}. Most notably, A* Prune~\cite{Liu_Ramakrishnan_2001} is a multi-metric adaptation\footnote{This adaptation is exact, \ie not a heuristic, as the estimated cost underestimates the actual distance towards the destination.} of A* relying on a PQ where paths known to be unfeasible are pruned. 
Two-phase methods~\cite{Raith_Ehrgott_2009} first find paths
lying on the convex hull of the Pareto front through multiple Dijkstra runs,
before finding the remaining non-dominated path through implicit enumerations.

Finally, \samcra~\cite{SAMCRA} is a popular and well-known multi-constrained path algorithm.
Similarly to other Dijkstra-based algorithms, \samcra relies on a PQ
to explore the graph but instead of the traditional lexicographical ordering, it relies on non-linear cost aggregation.
Among feasible paths (others are natively ignored) it first considers the one that minimizes its maximum distance to the multiple constraints. Such a path ranking to deal with the PQ is supposed to increase its performance with respect to other PQ organizations.

\begin{table*}[!ht]
    \centering
    \caption{Qualitative summary of a representative subset of DCLC-compatible algorithms showcasing their practicality, exactitude, and performance. In the \textit{Practical Features} column, the green check-mark indicates whether the algorithm supports the corresponding feature (while the red cross denotes the opposite). In the \textit{Exactitude vs Performance} column, the two subcolumns associated which each three scenarios show how the latter impact $(i)$ the exactitude (exact, strong guarantees, no guarantees) and $(ii)$ the performance of the algorithm (polynomial time or not). While the orange tilde denotes strong guarantees in terms of exactitude, green check-marks (and red crosses respectively) either indicate exact results (no guarantees resp.) or polynomial-time execution (exponential at worst resp.) for performance. For both subcolumns \textit{Bounded Pareto Front} and \textit{Coarse Metric}, we consider the case where their spreading is polynomial with respect to the number of vertices in the input graph (and as such predictable in the design/calibration of the algorithm).
    }
    \label{tab:relatedtable}
    \begin{tabular}{cccccccccc}
    \toprule
              Algorithms   & \multicolumn{3}{c}{Practical Features} & \multicolumn{6}{c}{Exactitude vs Performance}   \\
              \cmidrule(r){2-4}\cmidrule(l){5-10}
                  & Multi-Dest & SR   & Multi-thread & \multicolumn{2}{c}{Bounded} & \multicolumn{2}{c}{Coarse}  & \multicolumn{2}{c}{All} \\
                  & Single Run & Ready & Ready & \multicolumn{2}{c}{Pareto Front} & \multicolumn{2}{c}{Metric} & \multicolumn{2}{c}{Cases}  \\\midrule
    LARAC~\cite{larac} &  \darkredcross     & \darkredcross     & \darkredcross               & \darkredcross & \darkgreencheck & \darkredcross & \darkgreencheck & \darkredcross      & \darkgreencheck     \\
    LPH~\cite{EBFA01}   &  \darkgreencheck     & \darkredcross     & \darkredcross               & \darkgreencheck & \darkgreencheck & \orangetilde & \darkgreencheck & \darkredcross      & \darkgreencheck     \\
    HMCOP~\cite{Korkmaz_Krunz_2001}   &  \darkredcross     & \darkredcross     & \darkredcross  & \darkredcross & \darkgreencheck & \darkredcross & \darkgreencheck & \darkredcross      & \darkgreencheck     \\ \midrule
    HIPH ~\cite{HIPH}  &  \darkgreencheck    & \darkredcross     & \darkredcross               & \darkgreencheck & \darkgreencheck & \darkgreencheck      & \darkgreencheck & \orangetilde      & \darkgreencheck     \\
    Hassin~\cite{Hassin_1992} &  \darkredcross    & \darkredcross     & \darkredcross               & \orangetilde      & \darkgreencheck & \darkgreencheck     & \darkgreencheck & \orangetilde           & \darkgreencheck     \\
    Tsaggouris et al.~\cite{Tsaggouris_Zaroliagis_2009} & \darkgreencheck & \darkredcross & \darkredcross & \orangetilde & \darkgreencheck & \darkgreencheck & \darkgreencheck & \orangetilde & \darkgreencheck \\\midrule
    Raith et al.~\cite{Raith_Ehrgott_2009}  & \darkredcross & \darkredcross     & \darkredcross               & \darkgreencheck & \darkgreencheck      & \darkgreencheck & \darkgreencheck     & \darkgreencheck      & \darkredcross     \\
    A* Prune.~\cite{Liu_Ramakrishnan_2001}  & \darkredcross & \darkredcross     & \darkredcross               & \darkgreencheck & \darkgreencheck      & \darkgreencheck & \darkgreencheck     & \darkgreencheck      & \darkredcross     \\
    SAMCRA~\cite{SAMCRA} &  \darkgreencheck     & \darkredcross     & \darkredcross               & \darkgreencheck & \darkgreencheck      & \darkgreencheck & \darkgreencheck      & \darkgreencheck      & \darkredcross     \\ \midrule\midrule
    BEST2COP  & \darkgreencheck  & \darkgreencheck     & \darkgreencheck               & \darkgreencheck & \darkgreencheck  & \darkgreencheck & \darkgreencheck  & \orangetilde    & \darkgreencheck    \\ \bottomrule
    \end{tabular}
    \end{table*}


As we have seen so far, while many solutions exist, most possess certain drawbacks or lack certain features to reconcile both the practice and the theory.
Heuristics do not always allow to retrieve the existing paths enforcing strict SLAs, while exact solutions
are not able to guarantee a reasonable maximum running time when difficult instances arise, although both features are essential for real-life deployment.
On the other hand, FPTAS can provide both strong guarantees and a polynomial execution time.
However, they are often found in the field of operational research where, at best, possible networking applications and assumptions are discussed, but are not investigated.
Because of this, the deployment of the computed paths, with SR and its MSD in particular, is not taken into consideration. It is worth to note that the number of segments
is not a standard metric as it is not simply a weight assigned to each edge in the original graph (that is, without a specific construct, it requires to be computed on the fly for each visited path). Considering the latter can have a drastic impact on the performance of the algorithms not designed with this additional metric in mind.
In addition, not all the algorithms presented here and in Table~\ref{tab:relatedtable} are single-source multiple-destinations. Finally, none of these algorithms
evoke the possibility to leverage multi-threaded architectures, an increasingly important feature as such computations now tend to be performed by dedicated Path Computation Elements or even in the cloud.

Our contribution, \bestcop, aims to close this gap by mixing the best existing features (such as providing both a limited execution time and strong guarantees in terms of precision in any cases) and adapt them for a practical modern usage in IP networks deploying SR.
Table~\ref{tab:relatedtable} summarizes some key features of a representative subset
of the related work. Similarly to FPTAS, \bestcop rounds one of the metrics of the graph. However, conversely to most algorithms,
\bestcop does not sacrifice accuracy of the cost metric, but of the measured delay. Because of the native inaccuracy of delay measurements (and the
arbitrary nature of its constraint), this does not prevent \bestcop from being technically exact in most practical cases.
In addition (and akin to ~\cite{HIPH}), \bestcop can easily be tuned to remain exact on all simple instances with a bounded Pareto front regardless of the accuracy of the metrics. Thus, \bestcop can claim to return exact solutions in most scenarios and, at worst, ensure strict guarantees in others (for theoretical exponential instances). In all cases, \bestcop possesses a pseudo-polynomial worst-case time complexity. \bestcop
was designed while keeping the path deployment aspect of the problem in mind. A single run allows to find all DCLC paths (and many variants as we will see later) to all destinations. The MSD constraint related to SR is taken into account
natively. As a result, paths requiring more than MSD segments are excluded from the exploration space. The outer loop of \bestcop can be easily parallelized, leading to a non-negligible reduction in the execution time.
In Sec.~\ref{sec:evaluation}, relying on a performance comparison between \bestcop and \samcra, we will show that \bestcop does not even need to rely on multi-threading to provide lower computing times than \samcra while returning the exact same solutions (even though we make \samcra benefit from advantageous methods explained thoroughly in the remainder of this paper).
This result is particularly interesting as it remains true even for simple IP network instances. This comparison also enables to evaluate Dijkstra-oriented solutions (\samcra) with respect to Bellman-Ford-oriented ones (\bestcop).


Last but not least, \bestcop has been adapted for multi-area networks and leverages the structures of the latter, allowing it to solve DCLC on very large ($\approx$ \num{105000} nodes) domains in one second. To the best of our knowledge, such large-scale experiments and results have neither been conducted nor achieved within SR domains~\footnote{Some elementary algorithms (such as Multi-constrained Dijkstra) and more intricate solutions exhibit impressive computing times on even more massive road networks~\cite{Hanusse_Ilcinkas_Lentz_2020}. However, such networks are less dense than IP ones (and with metrics that are also even more correlated).}.

\subsection{SR is Relevant for DCLC: MSD is not a Limit}~\label{ssec:nbseg}

Given the MSD constraint, one may question the choice of SR for deploying DCLC paths in practice.
Indeed, in some cases, in particular if the metrics are not aligned\footnote{The delay and the IGP costs in particular. Since node segments represent best IGP paths, the IGP cost and the number of segments will most likely be aligned by design}, constrained paths
may required more than MSD detours to satisfy a
stringent latency constraint.

\begin{figure}
    \centering
    \includegraphics[width=1\textwidth]{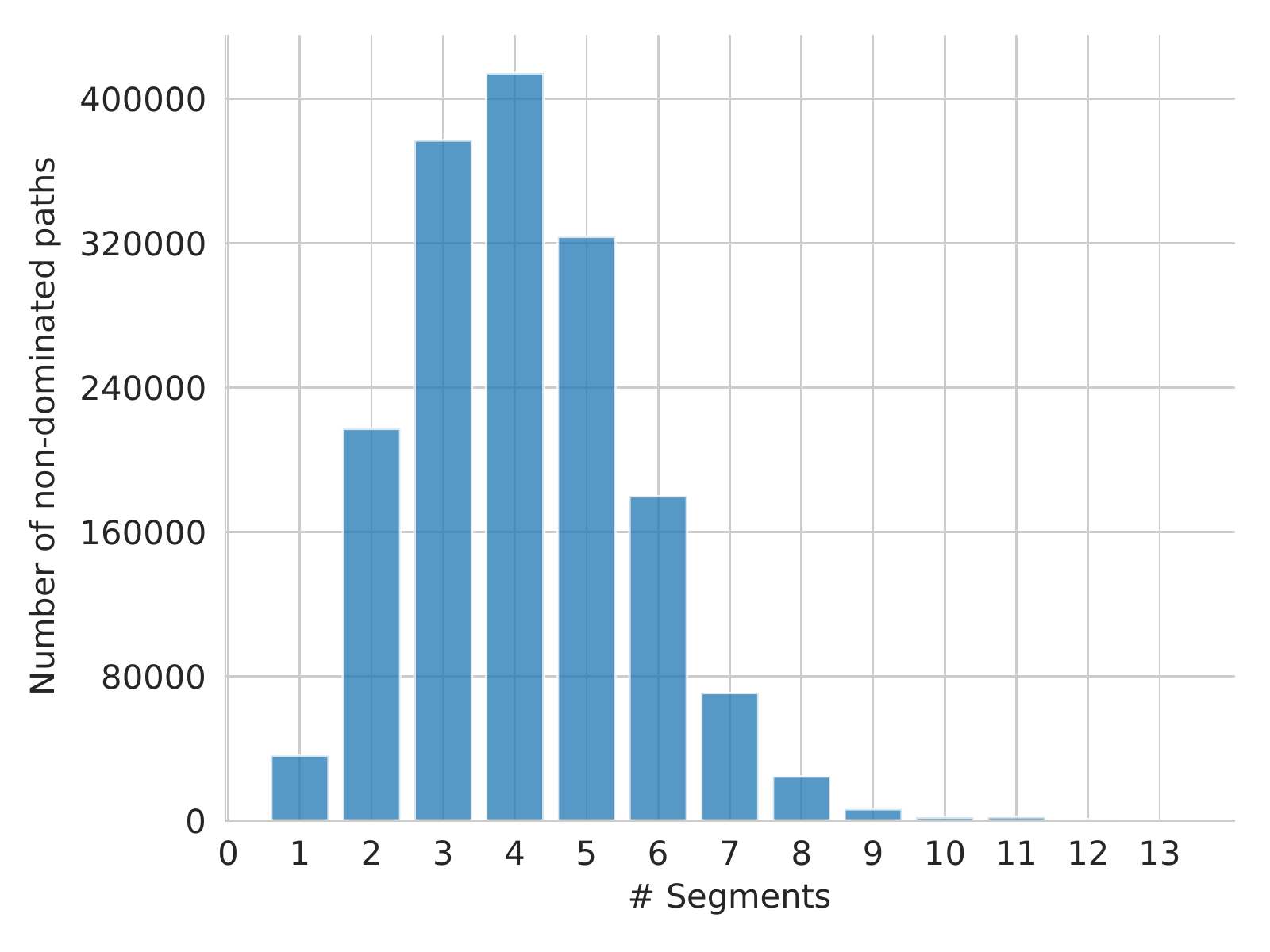}
    \caption{Required number of segments for all DCLC solutions, in a network of \num{45000} nodes generated by YARGG, with delay constraints of up to \num{100}ms.}
    \label{fig:nbseg}
\end{figure}

While it has been shown that few segments are required for most current SR usages (e.g. for TI-LFA \footnote{Topology Independent Loop Free Alternate} or when considering only one metric)~\cite{TheseClarenceFils,AubryPhd}, to the best of our knowledge, there is no similar study for our specific use-case, i.e. massive scale networks with two valuation functions (delay and IGP cost).
This is probably one of the most exciting challenge for SR as DCLC is a complex application.
However, since such massive-scale computer network topologies are not available publicly, we rely on our own topology generator whose detailed description is available in section \ref{sec:generator}.
These topologies follow a standard OSPF-like area division.
and both metrics (delay and IGP cost) follow a realistic pattern. For this analysis, we opt for a worst-case graph having $\approx~$\num{45000} nodes and $\approx~$\num{92000} edges scattered in \num{140} areas.

For this analysis, we keep track, for each destination, of all the solutions solving DCLC for all delay constraints up to 100ms, and extract the necessary number of segments.
In other words, we show the number of segments required to encode all non-dominated (and thus practically relevant for some given constraint) paths, considering all delay constraints up to 100ms.
The results are shown in Fig.~\ref{fig:nbseg}.



\red{One can see that most paths require less than 10 segments,
meaning that performant hardware should be able to
deploy most DCLC paths. However, some corner-cases requiring more than 10 segments do exist, probably arising
from stringent delay constraints. In addition, less performant hardware (\eg with MSD $\approx 5$), while able to deploy the majority of DCLC paths, can not deploy \textit{any} DCLC path. Note that several mechanisms exist to bypass this limit. Flexible Algorithms~\cite{I-D.ietf-lsr-flex-algo} allows to compute shortest paths and create segments with other metrics (\eg the delay). Binding segments~\cite{RFC8402} allows to "compress" a segment list in a single segment, which is uncompressed when popped from the list. However, both techniques increases the message exchange, number of states to maintain, and overall complexity. Their usage should thus be limited to a few corner cases.}

\red{Consequently, our analysis exhibits two main points. First, SR is appropriate to deploy TE paths. Indeed, the majority of DCLC paths should be deployable within the MSD constraint, if not all when using performant hardware. Second, since there may however exist DCLC paths requiring more than MSD segments, this limit \text{must} be considered to compute feasible paths correctly. Otherwise, a single non-feasible path dominating feasible ones is enough to lead to an incorrect algorithm. The underlying path computation algorithm must then efficiently consider delay, IGP cost and the number of segments to ensure its correctness.}

\section{BEST2COP(E): Efficient Data Structures and Algorithms for Massive Scale Networks}\label{sec:best2cop}

This section presents our contributions. We introduce and define preliminary notations and concepts used to design \bestcop,
before describing the data structures used by our algorithm.
In section \ref{sec:flat},  we describe our algorithm, \bestcop,
and show how we extend it for massive-scale networks divided in several areas in section~\ref{sec:areas}.

We have shown that SR seems indeed appropriate (as desired) for fine-grained delay-based TE. We thus aim to solve DCLC in the context of an ISP deploying SR, leading to the DCLC-SR problem that considers the IGP cost, the propagation delay, and the number of segments.

For readability purposes, we denote:
\begin{itemize}
  \item $M_0$ the metric referring to the number of segments, with the constraint $c_0$ = MSD applied to it;
  \item $M_1$ the delay metric, with a constraint $c_1$;
  \item $M_2$ the IGP metric being optimized.
\end{itemize}

Given a source $s$, \textbf{DCLC-SR} consists in finding, for \textit{all} destinations, a \textit{segment list} verifying
two constraints, $c_0$ and $c_1$, respectively on
the number of segments ($M_0$) and the delay ($M_1$),
while optimizing the IGP distance ($M_2$). We denote this problem DCLC-SR$(s, c_0, c_1)$.
On Fig.~\ref{fig:DCLCGeant}, we would have DCLC-SR$(\mathit{Frankfurt}, 3, 8) \supset \mathit{Frankfurt}-\mathit{Budapest}-\mathit{Vienna}$. This DCLC path (shown in yellow in Fig~\ref{fig:DCLCGeant}), is indeed the best option to reach Vienna when considering an arbitrary delay constraint of 8ms. Since the best IGP path from Frankfurt to Vienna (the green one) does not go through Budapest, encoding this DCLC path requires at least one detour, i.e. one segment (here, a node segment instructing the packet to go through Budapest first).


To solve such a challenging problem, efficient data structures are required.
In the following, we first introduce the constructs we leverage and how we benefit from the inaccuracy of real delay measurements in particular.

\subsection{DCLC and True Measured Delays}\label{sec:true}

\subsubsection{Leveraging measurement inaccuracy}
As mentioned, DCLC is weakly NP-Hard, and can be solved exactly in pseudo-polynomial time. In other words, as long as either the cost of the delay possesses only a limited number of distinct values (i.e., paths can only take a limited number of distinct distances),  the Pareto front of the paths' distances is naturally bounded in size as well, making DCLC tractable and efficiently solvable\footnote{Metric $M_0$ is omitted for now as this trivial distance is only required for SR and discussed in details later. While dealing with a three-dimensional Pareto front seems more complex at first glance, we will show that SR eventually reduces the exploration space because its operational constraint is very tight in practice and easy to handle efficiently.}.
Such a metric thus has to be bounded and possess a coarse accuracy (i.e., be discrete).
Although this has little impact when solving DCLC in a theoretical
context, it can be strongly leveraged to solve DCLC efficiently thanks to the characteristics of real ISP networks.

We argue that the metrics of real ISP networks do indeed possess a limited number of distinct values.
Although \bestcop can be adapted to fit any metric, we argue that $M_1$, the propagation delay, is the most appropriate one.
Indeed, IGP costs depend on each operators' configurations. For example, while some may rely on few spaced weights, others may possess intricate weight systems where small differences in weights may have an impact. Thus, bounding the size of the Pareto front based on the IGP costs is not only operator-dependant, but might still result in a very large front.

On the other hand, the delay (i) is likely strongly bounded,
and (ii) can be handled as if having a coarse accuracy in practice.
For TE paths, the delay constraint is likely to be very strict (10ms or less). Second, while the delay of a path is generally represented by a precise number in memory, the actual accuracy, \ie the trueness $t$ of the measured delay is much coarser due to technical challenges~\cite{RFC7679, RFC2681}.
In addition, delay constraints are usually formulated at the millisecond granularity with a tolerance margin, meaning that some loss of information is acceptable.

Thus, floating numbers representing the delays can be truncated to integers, e.g., taking 0.1ms as unit.
This allows to easily bound the number of possible non-dominated distances to $c_1 \times \gamma$, with $\gamma$ being the desired level of accuracy of $M_1$ (the inverse of the unit of the delay grain, here 0.1ms). 
For example, with $c_1 = 100$ms and a delay grain of 0.1ms ($\gamma = \frac{1}{0.1} = 10$), we have only $1000$ distinct (truncated) non-dominated pairs of distances to track at worst. This leads to a predictable and bounded Pareto front.
One can then store non-dominated distances within a static array, indexed on the $M_1$-distance (as there can only be one non-dominated couple of distances $(M_1, M_2)$ for a given $M_1$-distance).

 In the remaining of the paper, $\Gamma$ denotes the size allocated in memory for this Pareto front array (i.e., $\Gamma = c_1 \times \gamma$). When $t$, i.e., the real level of accuracy, is lower (or equal) than $\gamma$, the stored delay can be considered to be exact. More precisely, it is discretized but with no loss of relevant information.
When $t$ is too high, one can choose $\gamma$ such that $\gamma < t$,
to keep $\Gamma$ at a manageable value. In this case, some relevant information can be lost, as the discretization is too coarse. While this sacrifices the exactitude of the solution
(to the advantage of computation time), our algorithm is still able to provide predictable guarantees
in such cases (\ie a bounded error margin on the delay constraint). This is further discussed in Section~\ref{sec:b2coplimits}.

\red{\subsubsection{Fine, but which delay?}
Referring to a path's delay may be ambiguous. Indeed, this characteristic is not monolithic. The total delay is mainly composed of the propagation delay and the queuing delay. Both delays may play an important part in the overall latency, though none can be stated to be the main factor~\cite{end2end}. Although the propagation delay is stable, the queuing delay may vary depending on the traffic load. However, in order to compute TE paths, the delay metric must be advertised (usually within the IGP itself). For this reason, it is strongly recommended
to use a stable estimate of the delay, as varying delay estimations may lead to frequent re-computations, control-plane message exchanges, and fluctuating traffic distribution~\cite{RFC7471,RFC8570}.}

\red{For this reason, we use the propagation delay, as recommended in ~\cite{RFC7471,RFC8570}. The latter is usually measured through the use of a priority queue, ignoring so the queuing delay. Its value is deduced as a minimum from a sampling window, increasing so its stability~\cite{perfmes}. Using this delay not only makes our solution practical (as we rely on existing measurements and respect protocol-related constraints), but is actually pertinent in our case. In practice, flows benefitting from DCLC paths benefit from a queue with high priority and experience negligible queueing delays. In addition, the amount of traffic generated by such premium interactive flows can be controlled to remain small enough if it is not limited by design. Consequently, not only is there no competition between premium flows and best-effort traffic, but these flows do not generate enough traffic to lead to significant competition between themselves. Thus, the experienced delay is actually agnostic of the traffic load for our use-case, making the propagation delay a relevant estimate.
Consequently, we use the discretized propagation delay, enabling both practical deployment and the limitation of the number of non-dominated distances, within our structure used to encompass Segment Routing natively, the SR graph.
}

\subsection{The SR Graph and 2COP}\label{sec:sr-graph}

To solve DCLC-SR efficiently, as well as its comprehensive generalization, 2COP, we rely on a specific construct used to encompass SR, the delay, and the IGP cost: the \textit{multi-metric SR graph}.

\subsubsection{Turning the Physical Graph into a Native SR Representation}
\label{ssec:srgraph}
This construct represents the segments as edges to natively deal with the $M_0$ metric and its constraint, $c_0=\text{MSD}$.
The valuation of each edge depends on the distance of the path encoded by each segment.
While the weights of an adjacency segment are the weights of its associated local link, the weights of a node segment are the distances of the ECMP paths it encodes: the \red{(equal) IGP cost (\ie, M2-distance)}, and the lowest guaranteed delay (\ie, the M1-distance), \ie the worst delay among all ECMP paths.
Hence, computing paths on the SR graph is equivalent to combining stacks of segments (and the physical paths they encode), as stacks requiring $x$ segments are represented as paths of $x$ edges in the SR graph (agnostically to its actual length in the raw graph).
The SR graph can be built for all sources and destinations thanks to an All Pair Shortest Path (APSP) algorithm.
Note that this transformation is inherent to SR and leads to a complexity of $O(n (n\log(n) + m))$, for a raw graph having $n$ nodes and $m$ edges, with the best-known algorithms and data structures. This transformation (or rather, the underlying APSP computation) being required for any network deploying TE with SR (the complexity added by our multi-metric extension being
negligible), we do not consider it as part of our algorithm presented later.

This transformation is shown in Fig.~\ref{fig:examplesr}, which shows the SR counterpart of the raw graph provided in Fig.~\ref{fig:DCLCGeant}.
To describe this transformation more formally, let us denote $G = (V, E)$ the original graph, where $V$ and $E$ respectively refer to the set of vertices and edges. As $G$ can have multiple parallel links between a
pair of nodes $(u,v)$, we use $E(u,v)$ to denote all the direct links between
nodes $u$ and $v$.
Each link $(u,v)$ possesses two weights, its delay $w_1^{G}((u,v))$ and its IGP cost $w_2^{G}((u,v))$.
The delay and the IGP cost being additive metrics, the $M_1$ and $M_2$ distances of a path $p$ (denoted $d_1^G(p)$ and $d_2^G(p)$ respectively) are the sums of the weights of its edges.

From $G$, we create a transformed multigraph, the SR graph denoted $G' = (V,
E')$. While the set of nodes in $G'$ is the same as in $G$, the set of edges
differs because $E'$ encodes segments as edges representing either adjacency or node segments encoding respectively local physical link or sets of best IGP paths (with ECMP).
The $M_i$-weight of an edge in $G'$ is denoted $w_i^{G'}((u,v))$.
However, to alleviate further notations, we denote simply $d_i(p)$ the $M_i$ distance of a path in $G'$ instead of $d_i^{G'}(p)$.
Note that if $G$ is connected, then $G'$ is a complete graph thanks to node segments.

\red{\textbf{SR graph: Node segment encoding.}} A node segment, encoding the whole set $P_G(u,v)$ of ECMP best paths
between two nodes $u$ and $v$, is represented by exactly one edge in $E'(u,v)$. The $M_2$-weight $w_2^{G'}((u,v))$ of a node segment is the (equal) $M_2$-distance of $P_G(u,v)$.
Since, when using a node segment, packets may follow any of the ECMP paths,
we can only guarantee that the delay of the path will not exceed the
maximal delay out of all ECMP paths. Consequently, its $M_1$-weight $w_1^{G'}((u,v))$ is defined as the maximum $M_1$-distance among all the paths in $P_G(u,v)$.
Links representing node segments in $G'$ thus verify the following:
\[
\begin{array}{l}
    w_1^{G'}((u,v)) = \max_{p\in P_G(u,v)} d_1^G(p)\\
    w_2^{G'}((u,v)) = d_2^G(p)\quad \text{for any }p\in P_G(u,v)
\end{array}
\]

\red{\textbf{SR graph: Adjacency segment encoding.}} An adjacency segment corresponds to a link in the graph $G$ and is represented by an edge $(u_x,v)$ in $E'(u,v)$,
whose weights are the ones of its corresponding link in $G$, only if it is not dominated by the node segment $(u,v)_{G'}$ for the same pair of nodes, \ie if
$w_1^{G'}((u,v)) > w_1^{G}((u_x,v))$, or by any other non-dominated adjacency
segments $(u_y,v)$, \ie if $w_1^G((u_y,v)) > w_1^{G}((u_x,v))$ or
$w_2^G((u_y,v)) > w_2^{G}((u_x,v))$, where $(u_x, v)$ and $(u_y, v)$ are two different outgoing links of $u$ in $E(u,v)$\footnote{If two links have exactly the same weights, we only add one adjacency segment in $G'$}.\\

Fig.~\ref{fig:examplesr} illustrates the result of such a transformation: one can easily identify the three non-dominated paths between Frankfurt and Vienna, bearing the same colors as in Fig.~\ref{fig:DCLCGeant}.
The green path (\ie the best $M_2$ path) is encoded by a single node segment. The pink, direct path (\ie the best $M_1$ path) is encoded by an adjacency segment (the double line in Fig.~\ref{fig:examplesr}).
The yellow paths (the solution of DCLC-SR(Frankfurt, 3, 8) and an interesting tradeoff between $M_1$ and $M_2$) requires
an additional segment, in order to be routed through Budapest.
Note that in practice, the last segment is unnecessary if it is a node segment, as the packet will be routed towards its final IP destination through the best $M_2$ paths natively.

\begin{figure}
  \centering
  \includegraphics[scale=0.5]{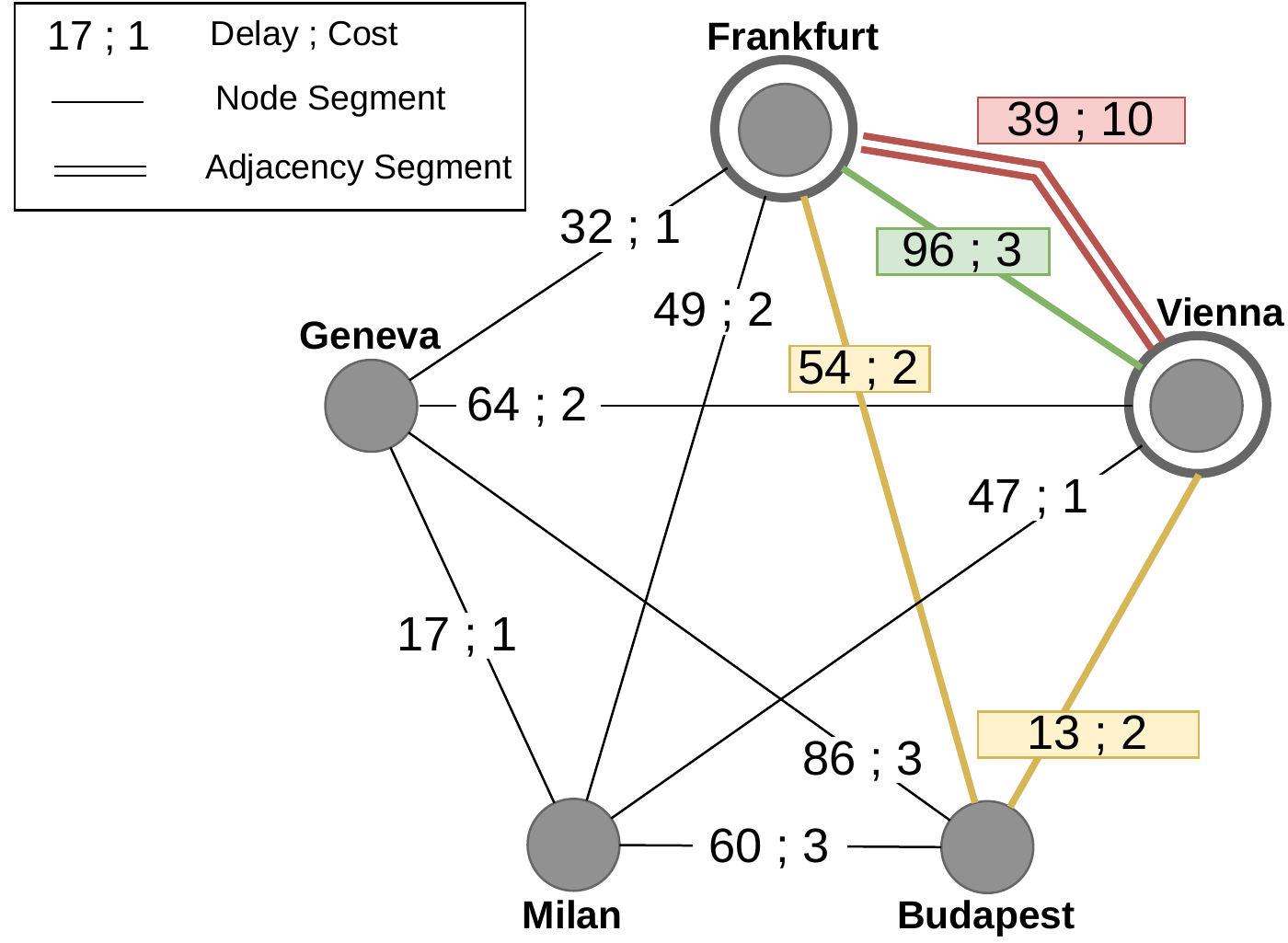}
  \caption{
  This Figure shows the network from Fig.~\ref{fig:DCLCGeant} translated into an SR graph.
  The SR graph encodes segments as edges. Plain edges represent node segments, \ie sets of ECMP paths. Double-lines are adjacency segments, here only $(\mathit{Frankfurt},\mathit{Vienna})$, and are visible only if they are not dominated by other segments. Colored edges refer to the paths highlighted in Fig.~\ref{fig:DCLCGeant}.}
  \label{fig:examplesr}
\end{figure}

\red{Our multi-metric SR graph (or equivalent constructs gathering the multi-metric all-pair shortest path data) is mandatory to easily consider the number of segments necessary to encode the paths being explored. However, its usage can differ in practice. We envision two modes which allow to consider this additional "off the graph" metric, using our SR Graph.}\\

\red{\textbf{Using the SR Graph to perform path conversions.} One of the two options is to run the path computation algorithm on the original topology, and convert the paths being explored to segment lists. Performing this conversion is however not trivial. One must return the minimal encoding of the given path (with respect to the number of segments) while correctly managing the (forced) path diversity brought by ECMP, which may exhibit heterogeneous delays. However, one can efficiently perform such conversion when relying on our SR graph. By
summarizing the relevant information (\ie the worst-case delay within ECMP paths), the SR Graph allows to easily consider the ECMP nature of SR within a multi-metric context. 
However, the segment metric $M_0$ is peculiar. Extending a path does not always imply an increase in the number of necessary segments. Furthermore, the number of segments required to encode two distinct paths may evolve differently, even when the latter are extended from the same node with the same edge. Because of these properties, the way to check paths for dominancy must be revised. This extended dominance check may lead to an increased number of paths to extend, and thus to a higher worst-case complexity. Further details regarding the conversion algorithm and the extended dominancy check can be found in ~\ref{app:LCA}.}\\

\red{\textbf{Using the SR Graph natively.} Another method is to run the path computation algorithm directly on the SR graph we described. 
Note that this forces the algorithm to run on a complete graph, which may significantly increase the overall complexity. However, the segment metric $M_0$, originally an "off the graph" metric with singular properties, becomes a standard graph metric, as it is now expressed by the number of edges that compose the paths (a path encoded by $x$ segments has $x$ edges within the SR Graph). This method also allows using standard, known algorithms as-is to solve the DCLC-SR problem.\\}

\red{When designing our algorithm, \bestcop, we use the second approach. Indeed, by using a bellman-ford-like exploration of the SR Graph, one can not only easily prune paths requiring more than $MSD$ segments, but also benefit from efficient Pareto front management and multi-threading. These various features allow \bestcop to efficiently solve not only DCLC-SR, but also 2COP, a more general and practically relevant problem regarding the computation of constrained paths within an SR domain. Note, however, that we will provide our competitor with both approaches to make the evaluation as fair as possible.}



\subsubsection{The 2COP Problem(s)}
\label{sec:2cop}
Solving DCLC-SR exactly requires, by definition, to visit the entirety of the Pareto front for all destinations.
However, although only some of these paths are DCLC-SR solutions for a given delay constraint, all paths visited during this exploration may be of some practical interest. In particular, some of them solve problems similar to DCLC but with different optimization strategies and constraints. By simply memorizing the explored paths (\ie storing the whole Pareto front within an efficient structure), one can solve a collection of practically relevant problems. For instance, one may want to obtain a segment path that minimizes the delay, another the IGP-cost, or the number of segments. Solving 2COP consists in finding, for all destinations, paths optimizing all three metrics independently, and respecting the given constraints.
We formalize this collection of problems as 2COP.
Solving 2COP enables more versatility in terms of optimization strategies and handles heterogeneous constraints for different destinations.
Simply put, while DCLC-SR is a one-to-many DCLC variant taking MSD into account, 2COP is more general as it includes all optimization variants.


With initial constraints $c_0, c_1, c_2$, \bestcop solves 2COP, \ie returns in a single run paths that satisfy smaller constraints $c_0', c_1', c_2'$ for any $c_i' < c_i$, $i = 0,1,2$, \red{offering more flexibility than simply returning the DCLC-SR solution}.

\textbf{Definition.}
  \textit{\textbf{2-Constrained Optimal Paths (2COP)}\\
  Let $f(M_j, c_0, c_1, c_2, s, d)$ be a function that returns a feasible segment path from s to d (if it exists),
  verifying all constraints $c_i, 0 \leq i \leq 2$ and optimizing $M_j,j \in {0,1,2}$. For a given source $s$ and given upper constraints $c_0, c_1, c_2$, we have}
   \begin{gather*} 2COP(s, c_0, c_1, c_2) =
   \bigcup_{\substack{
    \forall d \in V,\\
    \forall j \in \{0,1,2\},\\
    \forall c_j' \leq c_j
          }}
   f(M_j, c_0', c_1', c_2', s, d)
   \end{gather*}

Observe that, for any $s \in V$, DCLC-SR($s, c_0, c_1$) consists of the paths in $2COP(s, c_0, c_1, \infty)$ minimizing $M_2$.
Looking at Fig.~\ref{fig:examplesr}, we have two interesting examples (we rely on the first capital letter of the cities):
 \begin{gather*}
  f(M_2,3,70,\infty,F,V) = \makesrpath{{(F,B)}, {(B,V)}}~(67,4)\\
  f(M_1,3,\Gamma,\infty,G,B) = \makesrpath{{(G,M)}, {(M,B)}}~(77,4)
 \end{gather*}
In the second example, recall that the M1-distances are truncated to obtain integer values and $\Gamma$ is the maximum $c_1$ constraint we consider (multiplied by $\gamma$).

When the delay accuracy allows to reduce the problem's complexity sufficiently, \bestcop can solve exactly any of the variants within 2COP \red{and return any desired output of the image of $f$}.


In Sec.~\ref{sec:b2coplimits}, we detail how we can handle each 2COP variant with guarantees when the delay accuracy is too high to provide exact solutions while remaining efficient.
Solving 2COP can be implemented as efficiently as solving only DCLC-SR.

\subsection{Our Core Algorithm for Flat Networks}
\label{sec:flat}

In this section, we describe \bestcop, our algorithm efficiently solving 2COP (and so DCLC-SR).
Its implementation is available online\footnote{\url{https://github.com/talfroy/BEST2COP}}.
\red{Although \bestcop was already described in our previous contribution, a complete and far more detailed algorithmic description can be found in ~\ref{app:b2cop} for the interested reader.}

Akin to the SR graph computation, \bestcop can be run on a centralized controller but also by each router.
Its design is centered around two properties.
First, the graph exploration is performed so that paths requiring $i$ node segments are found at the $i^{\mathit{th}}+1$ iteration\footnote{Note that each adjacency segment translates to at least one necessary segment, two if they are not globally advertised and not subsequent.}, to natively tackle the MSD constraint.
Second, \bestcop's structure is easily parallelizable, allowing to benefit from multi-core architectures with low overhead.

\red{Simply put, at each iteration, \bestcop starts by extending the known paths by one segment (one edge in the SR graph) in a Bellman-Ford fashion (a not-in-place version to be accurate). Paths found during a given iteration are only checked loosely (and efficiently) for dominancy at first. This extension is performed in a parallel-friendly fashion that prevents data-races, allowing to easily parallelize our algorithm. \red{Only once} at the end of an iteration are the newly found paths filtered and thoroughly checked for dominancy, to reflect the new Pareto front. The remaining non-dominated paths are in turn extended at the next iteration. 
These steps only need to be performed $MSD \approx 10$ times, ignoring so all paths that are not deployable through SR.}
When our algorithm terminates, the results structure contains, for each segment number, all the distances of non dominated paths from the source towards all destinations.
\red{The interested reader may find further details regarding how 2COP solutions are extracted from the results structure in  Appendix~\ref{sec:best2cop}.}


\red{The good performance of \bestcop comes from several aspects. First, the fact that paths requiring more than MSD segments are natively excluded from the exploration space. Second, well-chosen data structures benefiting from the
limited accuracy of the delay measurements to limit the number of paths to extend. This allows to
manipulate arrays of fixed size, because the Pareto front of distances towards each node is limited to $\Gamma$ at each step (enabling very efficient read/write operations). Third, using a Bellman-Ford approach allows not only to easily parallelize our algorithm but also to perform lazy efficient update of the Pareto front. Indeed, a newly found path may only be extended at the \textit{next} iteration. Thus, we can efficiently extract the non-dominated paths from all paths discovered during the current iteration in a \text{single} pass, once at the end of the iteration. Conversely, other algorithms
tend to either check for dominancy whenever a path is discovered
(as the later may be re-extended immediately), or not bother to check for dominancy at all, \eg by relying solely on interval partitioning to limit the number of paths to extend.}



\subsection{For Massive Scale, Multi-Area Networks}
\label{sec:areas}

As shown in~\cite{nca2020}, this algorithms exhibits great performance on large-scale networks of up to \num{1000} nodes ($\approx 15$ms).
However, since the design of \bestcop implies a dominant factor of $|V|^2$ in term of time complexity\footnote{The detailed complexity is given in section \ref{sec:complexitybis}} (the SR graph being complete), recent SR deployments with more than $\num{10000}$ nodes would not scale well enough.
The sheer scale of such networks, coupled with the inherent complexity of TE-related problems, makes
2COP very challenging if not impossible to practically compute at first glance.
In fact, even \bestcop originally exceeds 20s when dealing with $\approx$\num{15000} nodes.
As we will see in the evaluations, this is much worse with concurrent options.

In this section, we extend \bestcop in order to deal efficiently with massive scale networks.
By leveraging the physical and logical partitioning usually performed in such networks, we manage to solve 2COP in $\approx 1$s even in networks of \num{100000} nodes.

\subsubsection{Scalabity in Massive Network \& Area decomposition}
The scalability issues in large-scale networks do not arise solely
when dealing with TE-related problems. Standard intra-domain routing protocols encounter issues past several thousands of nodes. Naive network design creates a large, unique failure domain resulting in numerous computations and message exchanges, as well as tedious management.
Consequently, networks are usually divided, both logically and physically, in \textit{areas}.
This notion exists in both major intra-domain routing protocols (OSPF and IS-IS). In the following, we consider the standard OSPF architecture and terminology but our solution can be adapted to fit any one of them

Areas can be seen as small, independent sub-networks (usually of around 100 - 1000 nodes at most). Within OSPF,
routers within an area maintain a comprehensive topological database of their own area only. Stub-areas are centered around the \textit{backbone}, or area 0. \textit{Area Border Routers}, or ABRs, possess an interface in both the backbone area and a stub area. Being at the intersection of two areas, they are in charge of sending a summary of the topological database (the best distance to each node) of one area to the other. There are usually at least two ABRs between two areas. We here (and in the evaluation) consider two ABRs, but the computations performed can be easily extended to manage more ABRs.
Summaries of a non-backbone area are sent through the backbone. Upon reception, ABRs inject the summary within their own area. In the end, all routers possess a detailed topological database of their own area and the best distances towards destinations outside of their own area.


\subsubsection{Leveraging Area Decomposition}

This partitioning creates obvious separators within the graph, the ABRs. Thanks to the latter, we can leverage this native partition in a \red{similar} divide-and-conquer approach, adapted to the computation of 2COP paths, by running \bestcop at the scale of the areas before exchanging and combining the results.
We do not only aim to reduce computation time, but also to keep the number and size of the exchanged messages manageable.



We now explain how we perform this computation in detail.
For readability purposes, we rely on the following notations:
$\mathcal{A}_x$ denotes area $x$.
$A_x$ denotes the ABR between the backbone and $\mathcal{A}_x$.
When necessary, we may distinguish the two ABR $A1_x$ and $A2_x$. Finally, b2cop($\mathcal{A}_x, s, d$) denotes the results (the non-dominated paths) from $s$ to $d$ within $\mathcal{A}_x$.
When $d$ is omitted, we consider all routers within $\mathcal{A}_x$ as destination. Figure~\ref{fig:cartesian product} illustrates a network with three areas, $x$, $y$ and 0, the backbone area.\\
\red{We here chose to detail a simple distributed and router-centric variant of our solution.
However, our solution may well be deployed in other ways, e.g. relying on controllers, or even a single one.
In such cases, the computation could be parallelized per area if needed. Such discussion is left for future work.\\}


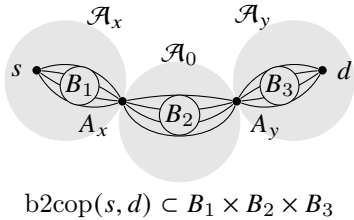
\begin{figure}\centering
  \begin{tikzpicture}[
    router/.style={black,fill, circle, inner sep=0.04cm}
  ]
    \definecolor{airecolor}{rgb}{0.9,0.9,0.9}
    \fill[airecolor] (0,0) circle (0.8);
    \fill[airecolor] (180-20:1.6) circle (0.8);
    \fill[airecolor] (20:1.6) circle (0.8);
    \node[] (AireX) at (180-55:1.7) {$\mathcal{A}_x$};
    \node[] (AireY) at (55:1.7) {$\mathcal{A}_y$};
    \node[] (Aire0) at (90:.9) {$\mathcal{A}_0$};

    \node[router,label={180:$s$}] (s) at (180-20:2) {};
    \node[router,label={-135:$A_x$}] (A1) at (180-20:0.8) {};
    \node[router,label={-45:$A_y$}] (A2) at (20:0.8) {};
    \node[router,label={0:$d$}] (d) at (20:2) {};

    \draw (s) edge[bend left=15] (A1);
    \draw (s) edge[bend left=40] (A1);
    \draw (s) edge[bend right=15] (A1);
    \draw (s) edge[bend right=40] (A1);

    \draw (A1) edge[bend left=10] (A2);
    \draw (A1) edge[bend right=70] (A2);
    \draw (A1) edge[bend right=15] (A2);
    \draw (A1) edge[bend right=40] (A2);

    \draw (A2) edge[bend left=15] (d);
    \draw (A2) edge[bend left=40] (d);
    \draw (A2) edge[bend right=15] (d);
    \draw (A2) edge[bend right=40] (d);

    \node[black, circle, draw, fill=airecolor, inner sep=0.25] (B1) at (180-20:1.4) {$B_1$};
    \node[black, circle, draw, fill=airecolor, inner sep=0.45] (B1) at (90:.09) {$B_2$};
    \node[black, circle, draw, fill=airecolor, inner sep=0.25] (B1) at (20:1.4) {$B_3$};

    \node[] (Aire0) at (-90:1.1) {b2cop$(s,d) \subset B_1\times B_2 \times B_3$};
  \end{tikzpicture}
  \caption{The set of solutions across areas is obtained from the cartesian product of the solutions in each area.}\label{fig:cartesian product}
\end{figure}

\textbf{Working at area scale.}
\red{Due to the area decomposition, routers do not possess the topological information to compute a full, complete SR graph of the whole network. Thus, we make routers only compute the SR graph of their own area(s).}
\red{Because exchanging the SR graphs themselves implies a large volume of information to share, we instead make the ABRs exchange their 2COP paths (i.e., the non dominated paths to all destinations of their areas)}  
since we limit their numbers to $\Gamma$ at worst. This exchange still provides enough information for all routers to compute all 2COP paths for every destination.




More formally, each ABR $A_x$ computes b2cop($\mathcal{A}_x, A_x$) and exchange the results with $A_y, \forall y \neq x$.
Areas being limited to a few hundreds routers on average, this computation is very efficient.
Note that ABRs also compute b2cop($\mathcal{A}_0, A_x$), but need not exchange it, as all ABRs perform this computation.
Exchanging the computed 3D Pareto front has a message complexity of $|V| \times c_0 \times \Gamma$ at worst in theory.
In practice, we expect both the size of Pareto fronts and the number of relevant destinations to consider to be fairly low ($<< \Gamma$ and $<< |V|$ resp.).
In the case of non-scalable Pareto fronts, one can opt for sending only part of them but at the cost of relaxing the guarantees brought by \bestcop.

After exchanging messages, any ABR $A_x$ should know the
non-dominated paths from itself to $A_y, \forall y \neq x$, and the non-dominated paths from $A_y$ to all nodes within $\mathcal{A}_y$. By combining this information, we can compute the non-dominated paths from $A_x$ to
all nodes within  $\mathcal{A}_y$, as we will now detail.\\



\textbf{Cartesian product.} Since ABRs act as separators within the graph, to reach a node within a given area $\mathcal{A}_y$, it is necessary to go through one of the corresponding ABRs $A_y$.
It thus implies that non-dominated paths to nodes within $\mathcal{A}_y$ from
$A_x$ can be found by combining bcop($\mathcal{A}_0, A_x, A_y$) with bcop($\mathcal{A}_y, A_y$).
In other words, by combining, with a simple cartesian product, the local non-dominated paths towards the ABRs of a given zone with the non-dominated paths from said ABRs to nodes within the corresponding distant areas, one obtains a superset of the non-dominated paths towards the destinations of the distant area. In practice, since several ABR can co-exist,
it is necessary to handle the respective non-dominated paths (bcop($\mathcal{A}_y, A1_y$) and bcop($\mathcal{A}_y, A2_y$)) with careful comparisons to avoid incorrect combinations.\\

\textbf{Post-processing and merging.} To ensure that the results obtained through the cartesian product aforementioned are correct, some post-processing is required.
When combining segment lists, the latter are simply concatenated.
More precisely, the resulting segment list necessarily possesses the following structure:
\makesrpath{{(u_0,u_1)}, \dots, {(u_i, A)}, {(A, v_0)}, \dots, {(v_{j-1}, v_j)}}, with $A$ denoting an ABR.
However, $A$ being a separator, it is likely that the best IGP path from $u_i$ to $v_0$ natively goes through $A$ without the need of an intermediary segment.
Thus, segments of the form \makesrpath{{(u_i, A)}, {(A, v_0)},} can often be replaced by a single segment $(u_i, v_0)$.
Such anomalies should be corrected, as an additional useless segment may render the path falsely unfeasible, even though it actually fits the MSD constraint.
This correction can be performed easily.
Let $A1$ be the separator, if $(u_i, A1)$ and $(A1, v_0)$ are
node segments, and all best IGP paths from $u_i$ to $v_0$ go through $A1$ (or possess the same cost and delay as the best IGP ones going through $A2$), the two node segments can be replaced by a single one.

This correction is performed quickly and relies solely on information available to the router (the local SR graph and the received distances summary).
Finally, after having performed and corrected the cartesian products for all the ABRs of the area, the latter are merged in a single Pareto front.


Once performed for all areas, an ABR $A_x$ now possesses all 2COP paths to all considered destinations within the network. These can then be sent to routers within $\mathcal{A}_x$, who will need to perform similar computations to compute non-dominated paths to all routers within a different area.
Note that the 2COP paths for each destination can be sent as things progress, so that routers can process such paths progressively (and in parallel) if needed.\\

\textbf{Summary.} By running \bestcop within each area, before exchanging and combining the results, one can find all non-dominated paths to each destination within a network of \num{100000} nodes in less than 900ms. The induced message complexity is manageable in practice and can be further tuned if required. Our method can be
adapted for controller-oriented deployments.


\subsection{A Limited Complexity with Strong Guarantees}
\label{sec:complexity}

\subsubsection{An Efficient Polynomial-Time Algorithm}
\label{sec:complexitybis}

\textbf{The flat \bestcop.} In the worst-case, for a given node $v$, there are up to $\mathit{degree}(v) \times \Gamma$
paths that can be extended towards it. Observe that $\mathit{degree}(v)$ is at least $|V|$ (because $G'$ is complete) and depends on how many parallel links $v$ has with its neighbors. With $L$ being the average number of links between two nodes in $G'$, \red{on average we thus have $\mathit{degree}(v) = |V|\times L \times \Gamma$ paths to extend to a given node, at worst.} These extensions are performed for each node $v$ and up to MSD times, leading to a complexity of $$O(c_0\cdot\Gamma\cdot |V|^2 \cdot L)$$
Using up to $|V|$ threads, one can greatly decrease the associated computation time.\\


\textbf{The Cartesian Product.} Its complexity is simply the size of the 2COP solution space squared, for each destination, thus at worst $O((c_0 \cdot \Gamma)^2 \cdot |V|)$. Note that we can reach a complexity of $O(c_0^2 \cdot \Gamma^2)$, again with the use of $|V|$ threads since each product is independent. This worst case is not expected in practice as metrics are usually mostly aligned to result in Pareto fronts whose maximal size is much smaller than $c_0 \cdot \Gamma$.\\ 

\textbf{Overall, \texttt{BEST2COPE} (multi-area) exhibits a complexity of}
\[
  O\left(c_0 \cdot \Gamma \cdot \left( c_0 \cdot\Gamma + L\cdot max_{\forall i \in [1..m]}(|V_i|) \right)\right)
  \]
  with $V_i$ denoting the set of nodes in each area $i$ ($m$ being their number) and the use of $|V|$ threads and sufficient CPU resources (\red{this bound is achievable ideally because the load is perfectly balanced and bottlenecks negligible}). Note that the cartesian product dominates this worst-case analysis as long as the product $V_i \cdot L $ remains small enough. However, with realistic weighted networks, we argue that the contribution of the Cartesian product is negligible in practice, so \texttt{BEST2COPE} is very scalable for real networking cases.

\subsubsection{What are the Guarantees one can expect when the Trueness exceeds the Accuracy, \ie if $t > \gamma$ ?}
\label{sec:b2coplimits}
~\newline
\indent If propagation delays are measured with a really high trueness (\eg with a delay grain of 1 $\mu$s or less), \bestcop (and so, BEST2COPE) can either remain exact but slower, or, on the contrary, rapidly produce approximated results.
In practice, if one prefers to favor performance by choosing a fixed discretization of the propagation delay (to keep the computing time reasonable rather than returning truly exact solutions), this may result in an array not accurate enough to store all non dominated delay values, \ie two solutions might end up in the same cell of such an array even though they are truly distinguishable. Nevertheless, we can still bound the margin errors, relatively or in absolute, regarding constraints or the optimization objective of the 2COP variant one aims to solve.

In theory, note that while no exact solutions remain tractable if the trueness of measured delays is arbitrarily high (for worst-case DCLC instances), it is possible to set these error margins to extremely small values with enough CPU power. If $t < \gamma$, each iteration of our algorithm introduces an absolute error of at most $\frac{1}{\gamma}$ for the $M_1$ metric, \ie the size of one cell in our array (recall that $\gamma = \frac{\Gamma}{c_1}$ is the accuracy level and is the inverse of the delay grain of the static array used by \bestcop{}). So our algorithm may miss an optimal constrained solution $p^*_d$ (for a destination $d$) only if there exists another solution $p_d$ such that $d_1(p_d) \geq d_1(p^*_d)$ but the $M_1$ distance of both solutions associated to the same integer (that is stored in the same cell of the \vardist array) \ie only if $d_1(p_d) \leq d_1(p^*_d) + \frac{c_0}{\gamma}$. In this case, we have $d_2(p_d) \leq d_2(p^*_d)$ because otherwise, $p^*_d$ would have been stored instead of $p_d$.
From this observation, depending on the minimized metric, \bestcop{} ensures the following guarantees.

If one aims to minimize $M_0$ or $M_2$ (\eg when solving DCLC), then \bestcop guarantees a solution $p_d$ that optimizes the given metric, but this solution might not satisfy the given delay constraint $c \leq c_1$. As an example, for DCLC-SR (optimizing $M_2$), we have

$$d_0(p_d) \leq c_0$$
$$d_1(p_d) < c + \frac{c_0}{\gamma}$$
$$d_2(p_d) \leq d_2(p^*_d)$$
With $p^*_d$ denoting the optimal constrained solution.
When minimizing $M_1$, the solution returned by \bestcop for a given destination $d$, $p_d$, will indeed verify the constraints on $M_0$ and $M_2$, and we have $d_1(p_d) < d_1(p^*_d) + \frac{c_0}{\gamma}$.
The induced absolute error of $c_0/\gamma$ regarding the delay of paths becomes negligible as the delay constraint increases.
If $c \approx c_1$, the latter translates to a small relative error of $c_0/\Gamma$. Conversely,
it becomes significant if $c << c_1$. When minimizing $M_0$ or $M_2$, it is thus recommended to set $c_1$ as low as possible regarding the relevant sub-constraint(s) $c \leq c_1$ if necessary. Similarly, to guarantee a limited relative error when minimizing $M_1$, it is worth running our algorithm with a small $c_1$ as we can have $d_1(p^*_d) << c_1$.
However, note that this later and specific objective (in practice less interesting than DCLC in particular) requires some a priori knowledge, either considering the best delay path without any $c_2$ and $c_0$ constraints, or running twice \bestcop to get $d_1(p_d)$ as a first approximation to avoid set up $c_1$ blindly initially (here $c_1$ is not a real constraint, only $c_2$ and $c_0$ apply as bounds of the problem, $c_1$ just represents the absolute size of our array and, as such, the accuracy one can achieve).



Even though \bestcop exhibits strong and tunable guarantees, it may not return exact solutions once two paths end up in the same delay cell, which may happen even with simple instances exhibiting a limited Pareto front.
Fortunately, a slight tweak in the implementation is sufficient to ensure exact solutions for such instances.
Keeping the original accuracy of M$_1$ distances, one can rely on truncated delays only to find the cell of each distance.
Then, one possible option consists of storing up to $k$ distinct distances in each cell\footnote{In practice, note that several implementation variants are possible whose one consists of using the array only when the stored Pareto front exceeds a certain threshold. Moreover, $k$ can be set up at a global scale shared for all cells or even all destinations, instead of a static value per cell, to support heterogeneous cases more dynamically. These approaches were also evoked in ~\cite{HIPH}}.
Thus, some cells would form a miniature, undiscretized Pareto front of size $k$ when required.
This trivial modification allows the complexity to remain bounded and predictable: as long as there exists less than $k$ distances within a cell, the returned solution is exact.
Otherwise, the algorithm still enforces the aforementioned guarantees.
While this modification increases the number of paths we have to extend to $k\cdot \Gamma$ at worst, such cases are very unlikely to occur in average. Notably, our experiments show that 3D Pareto fronts for each destination contain usually less than $\approx 10$ elements at most on realistic topologies, meaning that a small $k$ would be sufficient in practice.
In summary, \bestcop is efficient and exact to deal with simple instances and/or when $t \geq \gamma$, while it provides approximated but bounded solutions for difficult instances if $t < \gamma$ to remain efficient and so scalable even with massive scale IP networks.

\section{Performance Evaluation}
\label{sec:evaluation}

In this section, we evaluate the computation time of our solution.
We start by evaluating \bestcop on various flat network instances, ranging from worst-case scenario to real topologies, and compare it to another existing approach based on the \textit{Dijkstra} algorithm, \samcra~\cite{SAMCRA}.
Then, after having introduced our multi-area topology generator, we evaluate the extended variant of our solution, \bestcope, on massive scale networks.
In the following, we consider our discretization to be exact (\ie $\Gamma$ is high enough to prevent loss of relevant information).

\red{To conduct our evaluations, we consider that:}
\begin{itemize}
 \item $c_0 = \mathit{MSD} = 10$, as it is close to the best hardware limit;

\item $L = 2$: \red{while some pairs of nodes may have more than two parallel links connecting them in $G$, we argue that, on average in $G'$, one can expect that the total number of links in $E'$ is lower than $2|V|^2$.}

 \item $\Gamma = 1000$, although this value is tunable to reflect the expected product trueness-constraint on $M_1$, we consider here a fixed delay grain of 0.1ms (so an accuracy level of $\gamma = 10$) regarding a maximal constraint $c_1 =$ 100ms. This $\Gamma$ limitation is realistic in practice and guarantees the efficiency of \bestcop even for large complex networks as it becomes negligible considering large $|V|$.

\end{itemize}

\red{Note that the delays fed to \samcra are \emph{also} discretized in the same fashion as for \bestcop, allowing the number of non-dominated paths that \samcra has to consider to be bounded and reduced.}
\red{In addition, as \samcra is not designed with the SR Graph in mind, it is difficult to know which of the two methods mentioned in Section~\ref{ssec:srgraph} is the most suited to consider the segment metric. Thus, we compare ourselves to both variants. First, we run \samcra on the fully-meshed SR Graph, which allows to use the \samcra algorithm nearly as-is. We call this variant \samcrasrg. Second, we implement our conversion algorithm, which allows to efficiently convert multi-metric paths to segment lists. This method requires however further modification of the \samcra algorithm, not only by adding the conversion algorithm but also by extending its dominancy checks (details can be found in Appendix~\ref{app:LCA}). We refer to this variant as \samcralca.\\}

All our experiments are performed on an Intel(R) Core(TM) i7-9700K CPU @ 3.60GHz $\times$ 8.

\subsection{Computing Time \& Comparisons for Flat Networks}
This section illustrates the performance of our algorithm \bestcop using three flat network scenarios.
In particular, we do not take advantage of any area decomposition to mitigate the computing time.

\red{As shown in \cite{nca2020} by forcing \bestcop to explore its full iteration space, our algorithm \emph{cannot} exceed 80s at worst on topologies of \num{1000} nodes. This upper bound can however be drastically reduced through the use of multi-threading, reaching a worst-case of $\approx 10s$ when relying on 8 threads, highlighting the parallel nature of our algorithm~\footnote{Additional experiments on a high-performance grid showed that \bestcop may reach a speed-up of 23 when running on 30 cores.}. Additional information can be found in Appendix~\ref{upper_bound} for the interested reader. }

\red{We will see that in practice, \bestcop is far from reaching these upper bounds, even on random networks. In the following, we will evaluate \bestcop and compare it to \samcra in two main scenarios: a real network with real link valuations, and random networks of up to \num{10000} nodes.\\}

\begin{figure*}[!ht]
    \centering
    \caption{Computation time of \bestcop and \samcra on various experiments. Although the results can be close when considering mono-threaded \bestcop and \samcra, our algorithm always outperforms its competitor when using multi-threading. In some cases, multi-threading is not even necessary}\label{fig:allEvals}
        \subfloat[Computation time of \bestcop (1, 2, 8 threads) \& \samcra on a real network with real link valuation.]{\includegraphics[width=0.32\linewidth]{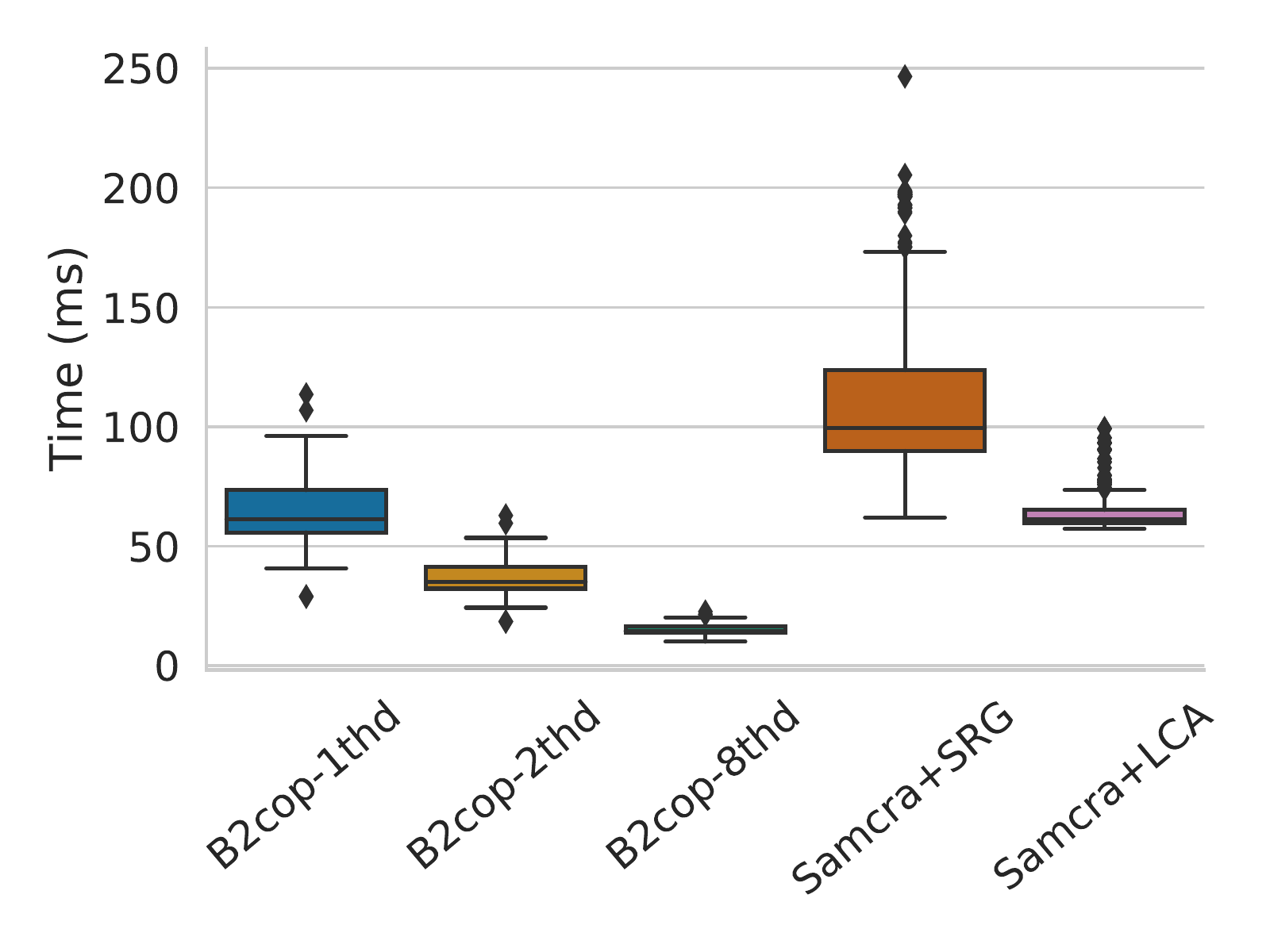}}\label{fig:geoSprint}
        \hspace{0.1em}
        \subfloat[Computation time of \bestcop \& \samcra on random networks with random valuation for $|V| = 100$ to \num{1000}]{\includegraphics[width=0.32\linewidth]{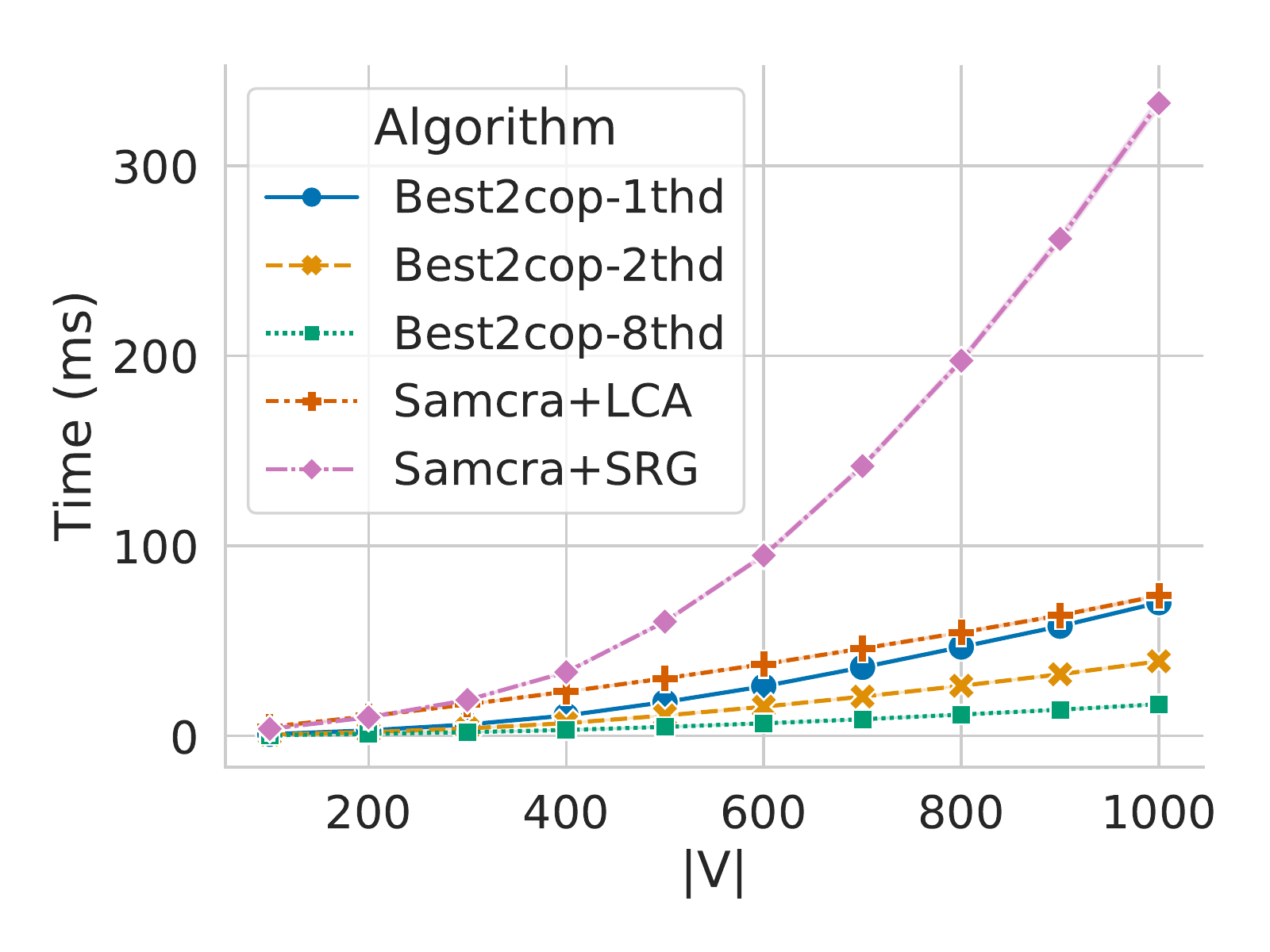}}\label{fig:1K}
        \hspace{0.1em}
        \subfloat[Computation time of \bestcop \& \samcra on random networks with random valuation for $|V| =$ \num{1000} to \num{10000}]{\includegraphics[width=0.32\linewidth]{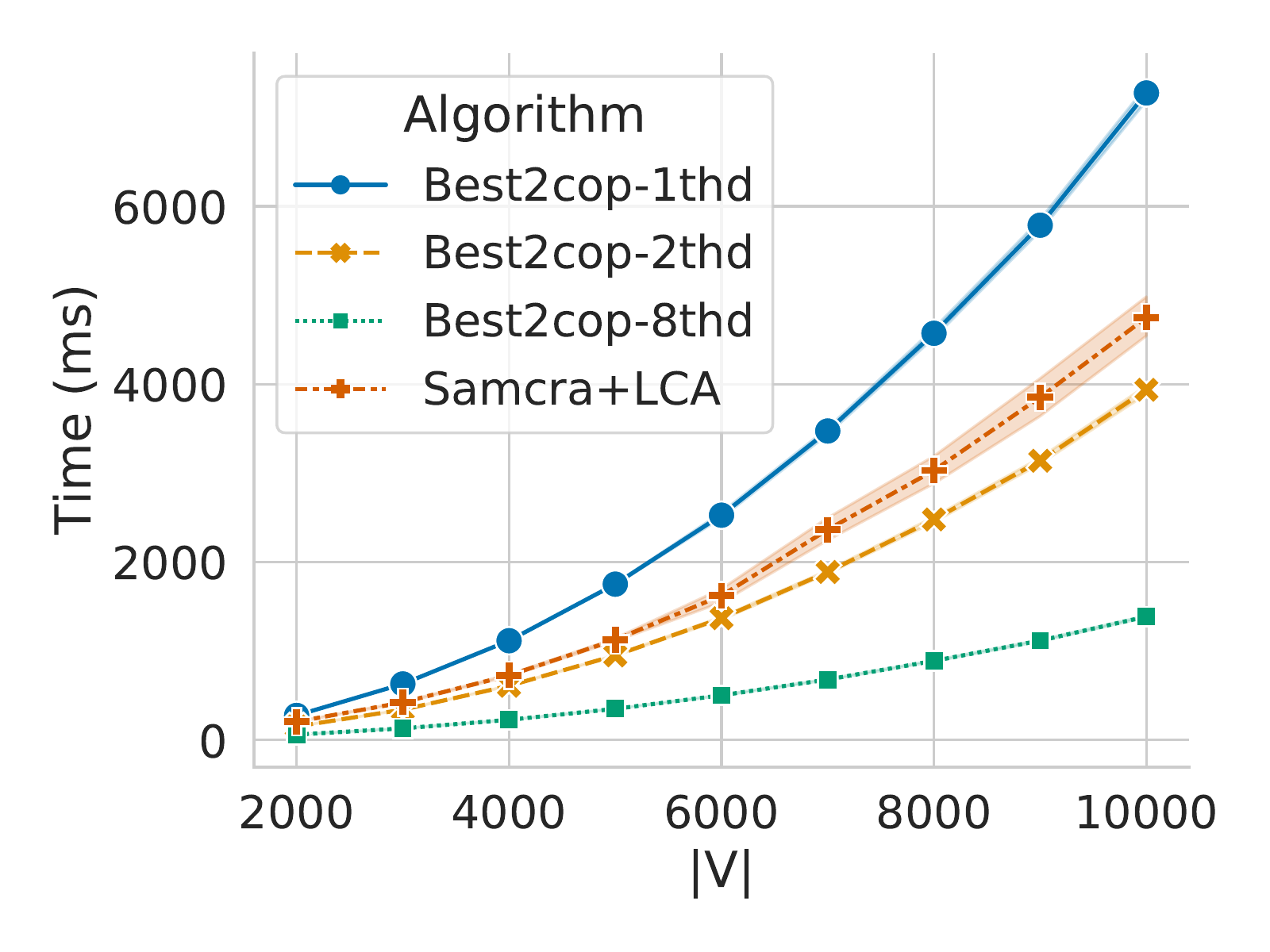}}\label{fig:10K}
\end{figure*}

\textbf{Real network.}
\red{We start by considering a real IP network topology. We use our largest available ISP topology, consisting of more than 1100 nodes and 4000 edges. \red{This topology describes the network of a Tier-1 operator and is not available to the public}~\footnote{While public topology datasets exist, these topologies are often too small for our use-case and/or do not possess any link valuation.}.
While the IGP costs of each link were available, we do not have their respective delays.
We thus infer delays thanks to the available geographical locations we do possess: we set the propagation delays as the orthodromic distances between the connected nodes divided by the speed of light, and run both algorithms on the obtained topology. The execution times are then shown in Fig.~\ref{fig:geoSprint}. \bestcop (1, 2, 8 threads) and \samcra (with LCA and SRG) are run for every node as source, resulting in the distributions showcased.}

\red{One can see that \samcrasrg (\ie \samcra run directly on the SR Graph) exhibits the worst execution times out of all the algorithms and variants presented, averaging at $100$ms, and reaching $250$ms at worst.
Interestingly, this shows that exploring the SR Graph itself may be detrimental to some algorithms (in particular priority-queue-based ones) due to its high density. Hence, algorithms not designed to take advantage of its features may fare better by exploring the original, sparser topology, and using the information within the SR Graph to compute the number of necessary segments to encode the paths being explored. This is visible on \samcralca computation times. Our construct, coupled with our conversion algorithm, allowed \samcralca to reach computation times very similar to the mono-threaded variant of \bestcop, with an average execution time of $\approx 60$ms. Note that \bestcop, which runs on the SR Graph itself, shows equivalent execution time when relying on a single thread. However, when relying on multiple threads, \bestcop outperforms its competitor in all runs, reaching a computation time of $\approx 25$ms at worst when using 8 threads, \ie three times faster than its competitor.}



These low execution times are not only due to the efficiency of the algorithms presented, but also to the realistic link valuations, which tend to be correlated in practice. In realistic cases, \bestcop can thus work with $\Gamma > 1000$ and so with a supported accuracy $t >> 0.1$ms (to deal with a micro-second grain) for small enough delay constraint (\ie, $<< 100$ms), while keeping the execution time in the hundreds of milliseconds.
One may notice that (almost) perfectly aligned metrics reduce the usefulness of any DCLC-like algorithm, but such metrics are not always aligned for all couples in practice (even with realistic cases, we observe that the average size of the 3D Pareto front is strictly greater than $1$, typically $\approx 4$). Our algorithm deals efficiently with easy cases and remains exact\footnote{Or at least near exact for difficult instances having both high trueness and exponential increasing Pareto fronts.} and efficient for more complex cases, \eg with random graphs.\\



\red{\textbf{Random networks. } The number of publicly available large topologies being limited, we continue our evaluation with random scenarios to assess the computation time of the aforementioned algorithms on a larger number of scenarios.}

\red{We generate raw connected graphs of $|V|$ nodes by using the Erdos-Rényi model. The generated topologies have a degree of $log(|V|)$. Both the delays and the 
IGP weights are picked uniformly at random. IGP weights are chosen within the interval $[1, 2^{32} / |V| / 10]$, to ensure that no paths possesses a cost higher than $2^{32}$. Delays are chosen within the interval $[0, 0.01 \times c_1]$, with $c_1 = 100ms$, to ensure that a high number of feasible paths exist.}

\red{We start by running \bestcop and \samcra for $|V|$ ranging from 100 to 1000 (with steps of 100). To account for the randomness of both valuation functions, we generate 30 differently weighted distinct topologies for each value of $|V|$. We run \bestcop and \samcra for 30 nodes selected as representative sources (randomly picked uniformly). Computing times are shown in Fig~\ref{fig:1K}.}

\red{While the computation times are slightly higher (due to the random valuations which lead to a higher number of non-dominated paths), the results are similar to the previous experiment. These results display more clearly that \samcra does not benefit from exploring the SR Graph. Indeed, on random networks, \samcrasrg is about 7 times slower than the other algorithms displayed. 
However, as on real networks, \samcralca shows results close (if not equal) to \bestcop execution time. Nevertheless, even on random networks, \bestcop remains three times faster than its competitor when relying on 8 threads.}

\red{Interestingly, \bestcop mono-threaded and \samcralca computation times get closer as $|V|$ increases. Thus, we continue our comparison on networks of \num{2000} to \num{10000} nodes. Given the long computation times of \samcrasrg, we here only consider \samcralca. The results are shown in Fig.~\ref{fig:10K}.
On such networks, \bestcop (mono-threaded) exhibits an execution time of $7s$, while \samcralca remains under $5s$. The quadratic complexity of \bestcop (whose main factor is $|V|^2$) is here clearly visible. 
\samcralca exhibits a less steep growth. However, when relying on multiple-thread, \bestcop remains far more efficient. While two threads already allow to reach an execution time slightly lower than \samcra (4s), 8 threads allow \bestcop to remain $\approx 3.3$ times faster than its competitor.} \\

\red{\textbf{Summary.} The way to use the SR Graph has a high impact on the underlying algorithm. As the SR Graph is not at the core of 
\samcra's design, exploring the latter lead to high execution time due to its density. However, adding our conversion algorithm (which relies on the SR Graph data) within \samcra allowed the latter to reach competitive execution times while solving 2COP. \bestcop, which explores the SR Graph directly, exhibits an execution similar to \samcralca when relying on a single thread. When using multi-threading, \bestcop clearly outperforms its competitor in all scenarios.\\}

In any case, it is interesting to note that even \bestcop takes more than one second on networks of \num{10000} nodes. To showcase the performance of our contribution on massive-scale networks, we now evaluate the execution time of its extension, \bestcope, which supports and leverages OSPF-like area division. This version is adapted to tackle TE problems in massive-scale hierarchical networks. In the following, we only consider our approach. Indeed, the latter showcased better performance that \samcra even without relying on multiple threads when considering a topology size $|V| < 1000$, which encompasses the size of standard OSPF areas.

However, before continuing analyzing the computing time results, we first introduce our generator for massive-scale, multi-areas, realistic network having two valuations (IGP cost \& delays).



\subsection{Massive Scale Topology Generation}
\label{sec:generator}

To the best of our knowledge, \red{although such networks exist in the wild,} there are no massive scale topologies
made \red{\textit{publicly}} available which exhibit IGP costs, delays, and area subdivision.
For example, the graphs available in the topology zoo (or sndlib) datasets do not exceed 700 nodes in general. Moreover, the ones for which the two metrics can be extracted, or at least inferred, are limited to less than 100 nodes.
Thus, at first glance, performing a practical massive-scale performance evaluation of \bestcope is highly challenging if not impossible.
There exist a few topology generators~\cite{IGEN, BRITE} able to generate networks of arbitrary size with realistic networking patterns, but specific requirements must be met to generate topologies onto which \bestcope can be evaluated, in particular the need for two metrics and the area decomposition.\\

\textbf{Topology generation requirements. }
First, the experimental topologies must be large, typically between \num{10000} and \num{100000} nodes.
Second, they must possess two valuation functions as realistic as possible, one for the IGP cost and the other modeling the delay.
Third, since the specific patterns exhibited by real networks impact the complexity of TE-related problems, the generated topologies must possess realistic structures (\eg with respect to redundancy in the face of failures in particular).
Finally, for our purposes, the topology must be composed of different areas centered around a core backbone, typically with two ABRs between each to avoid single points of failure. 
Since we do not know any generator addressing such requirements, we developed YARGG (Yet Another Realistic Graph Generator), a python tool(Code available online~\footnote{\url{https://github.com/JroLuttringer/YARGG}}) which allows one to evaluate its algorithm on massive-scale realistic IP networks. In the following, we describe the generation methods used to enforce the required characteristics.\\

\textbf{High-level structure.}
One of the popular ISP structure is the three-layers architecture~\cite{CampusNe63:online}, illustrated in Fig.~\ref{fig:area}. The \textit{access layer} provides end-users access to the communication service. Traffic is then aggregated in the \textit{aggregation layer}. Aggregation routers are connected to the
\textit{core routers} forming the last layer. The aggregation and access layers form an area, and
usually cover a specific geographical location. The core routers, the ABRs connecting the backbone other areas, and their links, form the backbone area that interconnect the stub areas, \ie the aggregation and access layers of the different geographical locations. Core routers are ABRs and belong both to a sub-area, per couple of 2 for redundancy.
Thus, while the access and aggregation layers usually follow standard structures and weight systems recommended by different network vendors, the backbone can vastly differ among different operators, depending on geographical constraints, population distribution, and pre-existing infrastructure. Taking these factors into account, YARGG generates large networks by following this 3-layer model, given a specific geographical location (\eg a given country).\\

\begin{figure}
    \centering
    \includegraphics[scale=0.15]{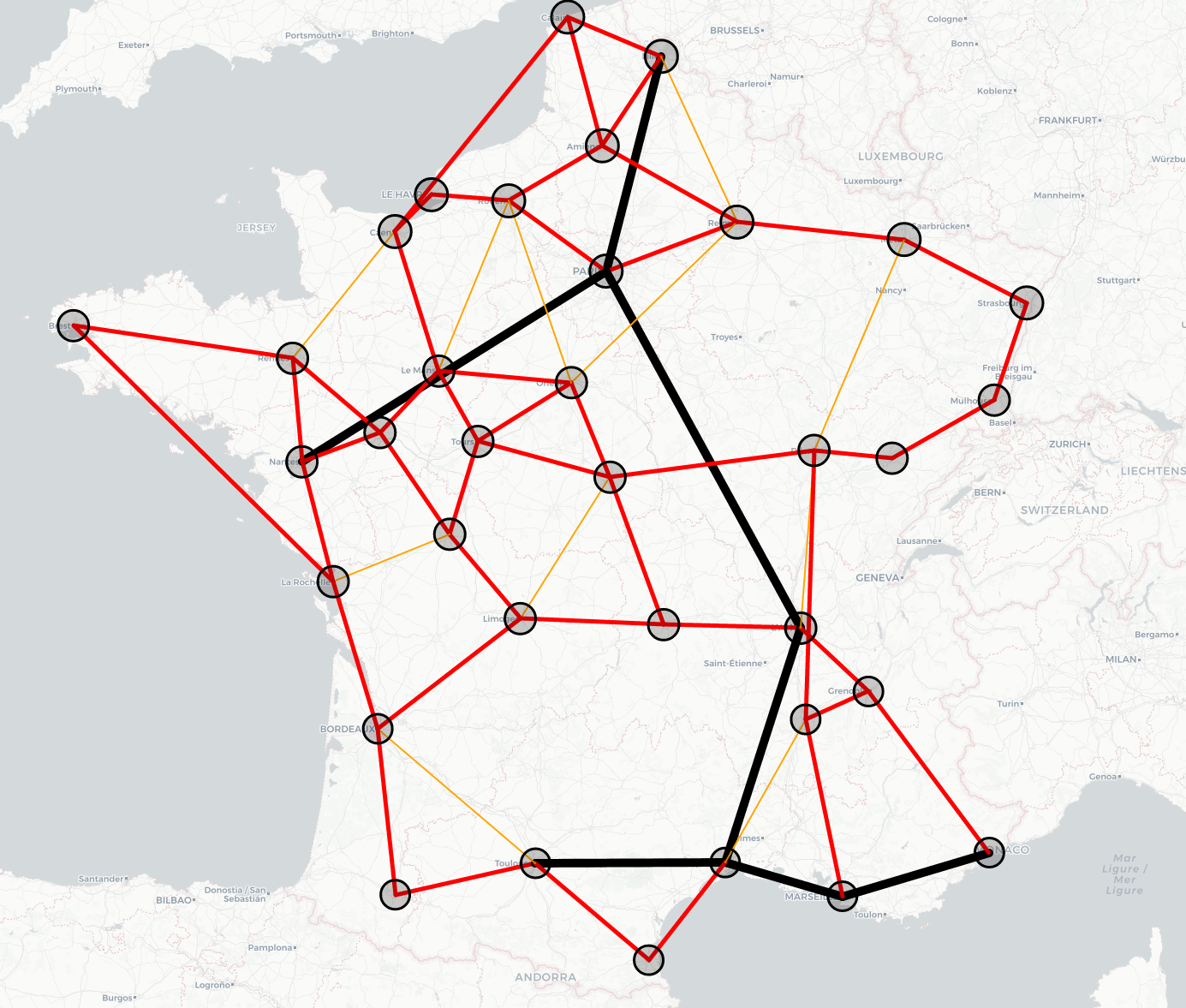}
    \caption{Core network (before step 5) generated by YARGG in France. While we consider the road distances, we represent the links in an abstract fashion for readability purposes. The color and width of the links represent their bandwidth (and thus their IGP costs).}
    \label{fig:coreFR}
\end{figure}

\textbf{Generating the core network and the areas. }
YARGG is a heuristic that generates the core network by taking the aforementioned considerations into account: existing infrastructures, population, and geographical constraints.
An example of a core network as generated by YARGG may be
seen in Fig.~\ref{fig:coreFR}.
Given a geographical location (e.g., a country or a continent), YARGG builds the structure of the core network by
 \begin{enumerate}
    \itemsep0em 
     \item \emph{Extracting the $x$ most populated cities in the area}. Close cities are merged in a single entity.\red{ The merge trigger value may change (the exact values used here can be found in ~\cite{topologies})}.
     \item \emph{Constructing a minimum Spanning-Tree covering all cities of the area}, using road distances as the metric. \red{Links between cities totaling more 
     than 30\% of the total population are normalized in order to be highly prioritized.\footnote{The parameters tuned may be easily modifier. For now, the latter are purely empirical.}}
     \item \emph{Removing articulation-points.} \red{YARGG picks one bi-connected component, and adds the smallest link (in terms of road distances) that 
     bridges this component with another. This process is repeated until no articulation point remains (\ie the topology is bi-connected).}
     \item \emph{Adding links increasing the connectivity and resilience for a limited cost.} \red{YARGG considers all links meeting certain criterias. The two cities/nodes must be closer than 20\% of the largest road distance. Their current degree must be lower than 4. The link, if added, should reduce the distance (and so, the delay) between the nodes by at least 25\%. Among these links, YARGG adds the one with the highest attractiveness, expressed as the sum of the distance reduction and the population of the cities (normalized).}
     \item \emph{Doubling the obtained topology.} The topology is doubled. There are now two nodes (/routers) per city. Links are added between the two routers of the same city, making the topology tri-connected.
 \end{enumerate}

The couple of routers located at each city within this generated backbone area become the ABRs between the backbone and their area, which is generated next.\\

\begin{figure}
    \centering
    \includegraphics[scale=0.6]{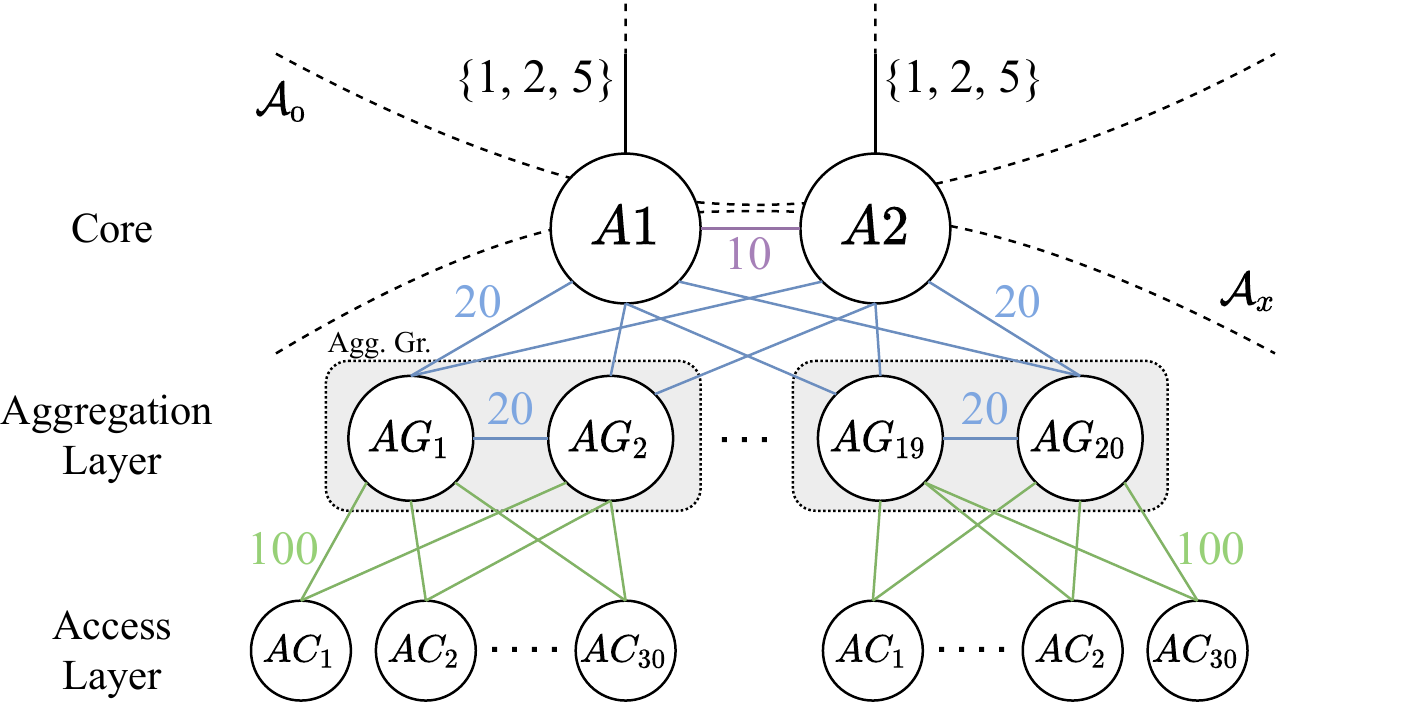}
    \caption{Weights and structures of an area generated by YARGG.}
    \label{fig:area}
\end{figure}

\textbf{Access \& aggregation layers. }
These last two layers make up a non-backbone area and span a reduced geographical area.
Thus, one access and one aggregation layer are located in each city considered by YARGG.
Several network equipment vendors recommend a hierarchical topology, such as the three-layer hierarchical model~\cite{Hierarch47:online}. An illustration can be seen in Fig.~\ref{fig:area}.
Simply put, there should be two core routers (the ABRs) at the given location (a city in YARGG's case).
Each core router is connected to all aggregation routers. For better resiliency, the aggregation layer is divided into aggregation groups, composed of two connected routers.
Finally, routers within an aggregation group are
connected to access-layer routers. To achieve areas of $\approx$ 300 nodes, we consider 30 access routers per aggregation group. This results in a large, dense, and realistic graph. \\

\textbf{Weights. } In the backbone, the weights generated by YARGG are straightforward.
The delays are extracted from the road distances between the cities, divided by $60\%$ the speed of light (close to the best performing fiber optic). The IGP cost is 1 for links between large cities since these links usually have a high bandwidth (in black in Fig.~\ref{fig:coreFR}), 2 for standard links, necessary to construct a tri-connected graph (added at step 3, in red in Fig.~\ref{fig:coreFR}), and 5 for links that are not mandatory, but that increase the overall connectivity (added at step 4, in orange in Fig.~\ref{fig:coreFR}).

Within an area, the IGP costs follow a set of realistic
constraints, according to two main principles:
(i) access routers should not be used to route traffic (except for the networks they serve), (ii) links between routers of the same hierarchical level (\eg between the two core routers or the two aggregation routers of a given aggregation group) should not be used, unless necessary (\eg multiple links or node failures). These simple principles lead to the
IGP costs exhibited in Fig.~\ref{fig:area}.
The delays are then chosen uniformly at random.
Since access routers and aggregation routers are close geographically, the delay of their links is chosen between $0.1$ and $0.3$ms. The delay between aggregation routers and core routers is chosen between $0.3$ms and the lowest backbone link delay.
Thus, links within an area necessarily possess a lesser delay than core links.\\

\textbf{Summary. } YARGG computes a large, realistic, and multi-area topology. The backbone spans a given geographical location and possesses simple IGP weights and realistic delays. Other areas follow a standard three-layer hierarchical model. Weights within a stub area are chosen according to a realistic set of usual ISP constraints. Delays, while chosen at random within such areas, remain consistent with what should be observed in practice.

\subsection{Computing Time for Massive Scale Multi-Areas Networks}
\begin{figure*}
    \centering
    \includegraphics[scale=0.45]{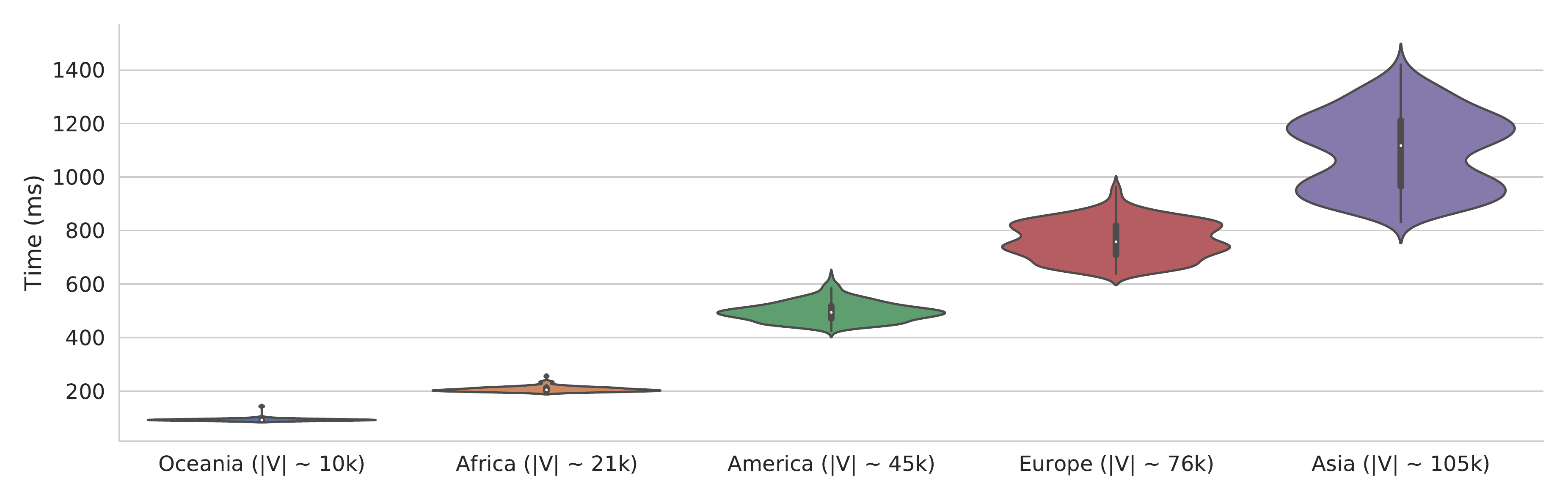}
    \caption{\bestcope computation time on 5 continent-wide topologies generated by YARGG.}
    \label{fig:bcope}
\end{figure*}

Using YARGG, we generate five massive scale, continent-wide topologies, and run \bestcop on each one of them. The topologies ranges
from \num{10000} to \num{100000} nodes. 
Each non-backbone area possess around \num{320} nodes. The topologies, their 
geographical representations and some of the associated network characteristics can be found online~\cite{topologies}.

We run \bestcop on each ABR as a source (around $|V| / 320 \times 2$ sources). The time corresponding to the message exchange of the computed Pareto front (step 2 of \bestcope) is not taken into consideration. Thus, 
the computation time showcased is the sum of the average time taken by ABRs to perform the preliminary intra-area \bestcop (and the distances to segment lists conversions) plus the time taken to perform the $|V| / 320 \times 2 - 2$ Cartesian products (for all other ABRs of all other areas).

Note that we consider an ABR as a source and not an intra-area destination.
In practice, the ABR would send the computed distances to the intra-area nodes, who in turn would have to perform a Cartesian product of these distances with its own distances to said ABR.
However, both the ABR and the intra-area node have to consider the same number of destinations ($|V|$), and the results computed by the ABR can be sent as they are generated (destination per destination), allowing both the ABR and the intra-area nodes to perform their Cartesian product at the same time. In addition, intra-area nodes may benefit from several optimizations regarding their Cartesian product, if the constraints of the desired paths are known (these optimizations will not be used nor detailed in this paper).
For these reasons, we argue that the time measured here, using an ABR as a source, is representative of the total actual time required, \ie the overall worst time for the last treated destination at each source.

The results of this experiment are shown in the violin plot of Fig.~\ref{fig:bcope}.
By leveraging the network structure, \bestcope exhibits very good performance despite the scale of the graph.
For \num{10000} nodes, \bestcope exhibits a time similar to the one taken by its flat variant for $|V| = 2000$.
Furthermore, \bestcope seems to scale linearly with the number of nodes, remaining always under one second for $|V| \approx \num{75000}$. Even once the network reaches a size of $\approx \num{100000}$, \bestcope is able to solve 2COP in less than one second for a non-negligible fraction of the sources, and never exceeds 1.5s.

Note that the times showcased here rely on a single thread.
While \bestcope's Cartesian product can be parallelized locally (both at the area and the destination scale), this parallelization hardly has any effect. This is explained by the fact that these individual computations are in fact fairly efficient, hence the overhead induced by the creation and management of threads is heavier than their workload. In addition, since \bestcop deals with very large topologies, some complex memory-related effects might be at play. Indeed, we notice these results to surprisingly vary depending on the underlying system, operating system, and architecture due to the differences in terms of memory management.

Thus, while massive scale deployments seem to a priori prevent the usage of fine-grained TE, their structures can be leveraged, making complex TE problems solvable in less than one second even for networks reaching $\num{100000}$ nodes. The computations performed for each 
area can also be distributed among different containers 
within the cloud, if handled by a controller.

\section{Conclusion}
While the overhead of MPLS-based solutions lead to a TE winter in the
past decade, Segment Routing marked its rebirth. In particular, SR enables the
deployment of a practical solution to the well-known DCLC problem. \red{Our algorithm,
\bestcop~\cite{nca2020} (Best Exact Segment Track for 2-Constrained Optimal Paths), iterates on the SR Graph to natively solve DCLC in SR domains with strong guarantees, through simple and efficient data structures and concepts.}


In this paper, we went several steps further with the following achievements:
\begin{itemize}
\item experimentally demonstrating that SR is a relevant technology to deploy DCLC paths;
\item for massive scale ISPs relying on area-subdivision, we extend \bestcop to \bestcope, partitioning 2COP into smaller sub-problems, to further reduce its overall complexity (time, memory and churn);
\item through extensive evaluations, relying on multi-threading and our own multi-metric/multi-areas network generator, we have shown that \bestcope is very efficient in practice. This was confirmed through a comparison with a relevant state-of-the-art algorithm, which benefited from a novel path to segment multi-metric conversion algorithm that we designed.
\end{itemize}

To the best of our knowledge, \bestcop is the first practically exact and efficient solution for 2COP within SR domains, making it the most practical candidate to be deployed for such a TE flavor in today ISPs.
It is able to solve 2COP on massive scale realistic networks having \num{100000} nodes in less than a second.
For large areas having thousands of routing devices, we have shown that \var{BEST2COP(-E)} can easily deal with random topologies while its competitors do not scale.
Finally, more advanced and flexible structures can be envisioned to deal with high trueness requirements, while deploying novel flex-algo strategies can help to mitigate the rare SR limit drawbacks.

\section*{Acknowledgment}
This work was partially supported by the French National Research Agency (ANR) project Nano-Net under contract ANR-18-CE25-0003.

\bibliographystyle{elsarticle-harv}
\bibliography{reference}
\newpage
\appendix
\section{The BEST2COP Algorithm: A detailed description}\label{app:b2cop}

This Appendix aims at describing \bestcop in a complete and detailed fashion. It also showcases the pseudo-code of our algorithm.

\textbf{\bestcop's main procedure} is shown in Alg.~\ref{algo:bestmcp}. 
The variable \varpfront, the end result returned by our algorithm, contains, for each iteration,
the Pareto front of the distances towards each node \varn. In other words,
\varpfront[i] contains, at the end of the $i^{\mathit{th}}$ iteration, all non-dominated $(M_1, M_2)$ distances of feasible paths towards each node \varn.

The variable \vardist is used to store, for each vertex, the
best $M_2$-distance found for each $M_1$-distance to each node.
Since the $M_1$-distance of any feasible path in $G'$ is bounded by $\Gamma$, we can store
these distances in a static array \vardist[\varv].
Note that during iteration $i$, \vardist will contain the Pareto front of the current iteration (non-dominated distances of $i$ segments) in addition to
distances that may be dominated. Keeping such paths in \vardist allows us to pre-filter paths before ultimately extracting the Pareto front of the current iteration from \vardist later on. This variable is used in conjunction with \varpfcand, a boolean array to remember
which distances within \vardist were found at the current iteration.

The variable \extendableIndex is a simple list that contains, at iteration $i$, all non-dominated distances discovered at iteration $i-1$.
More precisely, \extendableIndex is a list of tuples $(u, d\_list)$, where $d\_list$ is the list of the best-known paths towards $u$. The variable \extendableIndexbis is a temporary variable allowing to construct \extendableIndex.

After the initialization of the required data structures, the main loop starts. This loop is performed
MSD times, or until no feasible paths are left to extend.
For each node $v$, we extend the non-dominated distances found during the previous
iteration towards $v$ (originally (0,0) towards \texttt{src}). Extending paths in this fashion allows to easily parallelize the main \texttt{for} loop (\eg through a single \textit{pragma} Line~\ref{pragma}). Indeed, each thread can manage a different node $v$ towards which to extend the non-dominated paths contained within \extendableIndex. As threads will discover distances towards different nodes $v$ (written in turn in structures indexed on $v$), this prevents data-races. Note that in raw graphs, this method may lead to uneven workloads, as not all paths may be extendable towards any node $v$. However, since an SR graph is (at least) complete, any path may be extended towards any node $v$, leading to similar workloads among threads.

\begin{algorithm} 
\footnotesize
\caption{ BEST2COP(G',src)}\label{algo:bestmcp}
\varpfront~:= Array of size MSD \\
\ForAll{i $\in [0..MSD]$}{
    \varpfront[i] := Array of size $|V|$ of Empty Lists
    }
 \add(\varpfront[0][src], (0,0))\\

\vardist~:= Array of size $|V|$\\\label{line:init var all}
\ForAll{\varn $\in V$}{
   \vardist[\varn] := Array of size $\Gamma$\\
}
 \vardist[\var{src}][0] = (0,0) \\
optimal constrained
 \extendableIndex := Empty List (of Empty Lists)\\
\add(\extendableIndex, (\var{src}, [~(0,0)~]))\\

\extendableIndexbis := Array of size $|V|$ of Empty Lists\\

 \vari := 1, \varmaxd := 0\\
 \While{ \extendableIndex $\neq $ [~] \kwand \vari $\leq$ MSD}{
    {\color{purple}\texttt{\#pragma omp parallel for}}\label{pragma}\\
    \ForAll{\varv $\in V$}{

    \varpfcand~:= Array of size $\Gamma$\\
     \varnb, \varimax~:= \extend(v, \extendableIndex, \varpfcand, \vardist[v])\\
    \varmaxd = \kwmax(\varimax, \varmaxd)\\
  \tcp{How to iterate on \vardist to get new PF}
  \uIf{ \varnb log \varnb + \varnb + |\varpfront[i-1][\varv]| $<$ \varmaxd\label{choice}}
  {
   \var{d1\_it} := \mergesortd(\varpfront[\vari-1][\varv], \varpfcand)\label{merge}
  }
  \Else{\var{d1\_it} := [0\dots \varmaxd]}

    \extendableIndexbis[\varv] = [~]\\
    \tcp{Extract new PF from dist}
    \compute(\extendableIndexbis[\varv], \varpfront[\vari][\varv], \varpfcand, \var{d1\_it}, \vardist[v])
    }
    \tcp{Once each thread done, gather ext. paths}
    \extendableIndex = []\\ \label{unpar1}
    \ForAll{\varv $\in V~|~$ \extendableIndexbis[\varv] $\neq $ [~]}{
    \add (\extendableIndex, (\varv,\extendableIndexbis[\varv])) \\\label{unpar2}

    }

  \vari = \vari + 1\\
}
\Return \varpfront
\end{algorithm}

The routine \texttt{\textbf{ExtendPaths}}, detailed in Alg.~\ref{algo:extendsPaths}, takes the list of extendable paths, \ie non-dominated paths discovered at the previous iteration, and a node $v$. It then extends the extendable paths to $u$ further towards $v$.
The goal is to update \vardist[\varv] with new distances that may belong to the Pareto front.
Before being added to \vardist[\varv], extended distances go through a pre-filtering.
Indeed, the newly found distance to $v$ may be dominated or may be part of the Pareto front. While this check is performed thoroughly later, we can
already easily prune some paths: if the new paths to $v$ violate either constraint, there is no point in considering it. 
Furthermore, recall that \vardist stores, for all $\Gamma$ $M_1$-distances towards a node,
the best respective $M_2$-distance currently known. Thus, if the new $M_2$-distance is worst than the one previously stored in \vardist at the same $M_1$ index, this path is necessarily dominated and can be ignored. 
Otherwise, we add the distances to \vardist
and update \varpfcand to remember that a new distance which may be non-dominated was added during the current iteration. 
Note that \texttt{ExtendPaths} returns the number of paths updated within \vardist, as well as the highest $M_2$-distance found. This operation is performed for efficiency reasons detailed here.

Once returned, \vardist contains distances either dominated or not.
We thus need to extract the Pareto front of the current iteration.
This operation is performed in a lazy fashion once for all new distances (and not for each edge extension).
Since this Pareto front lies within \vardist, one can simply walk through \vardist by order of increasing $M_1$ distance from $0$ to the highest $M_1$ distance found yet
and filter all stored distances to get the Pareto front of the current iteration.
This may not be effective as most of the entries of \vardist may be empty.

However, the precise indexes of all active distances that need to be examined (to skip empty entries) can be constructed by merging and filtering the union of the current Pareto front and the new distances (\varpfcand).
Thus, if the sorting and merging of the corresponding distance indexes is less costly than walking through \vardist, the former method is performed in order to skip empty or useless entries. Otherwise, a simple walk-through is preferred. The merging of the $M_1$ distances of the Pareto front and new $M_1$ distances is here showcased at high-level (Line~\ref{merge}, Alg.~\ref{algo:bestmcp}). The usage of more subtle data structures in practice allows to perform this operation at the cost of a simple mergesort.

After the list of distance indexes to check and filter is computed, the actual Pareto front is extracted during the \texttt{\textbf{CptExtendablePaths}} procedure, as shown in Alg.~\ref{algo:ComputeExtendable}.
This routine checks whether paths of increasing $M_1$ distance do possess a better $M_2$ distance than the one before them. If so, the path is non-dominated and is added to the Pareto front,
as well as to the paths that are to be extended at the next iteration.
Finally, once each thread is terminated, \extendableIndexbis contains
$|V|$ lists of non-dominated distances towards each node. These lists are merged within \extendableIndex, to be extended at the next iteration.

Note that most approximations algorithms relying on interval partitioning or rounding do not
bother with dominance check. In other words, the structure they maintain is similar to
our \vardist : the best $M_2$ distance for each $M_1$ distance to a given node. The latter
may thus contain dominated paths which are considered and extended in future iterations.
In contrast, by maintaining the Pareto front efficiently, we ensure to consider the minimum
set of paths required to remain exact, and thus profit highly from small Pareto front.\\

\begin{algorithm} 
\footnotesize
 \caption{\footnotesize ExtendPaths(v, \extendableIndex, \varpfcand, \vardistv)}\label{algo:extendsPaths}

      \varimax = 0, \varnb = 0\\
      \ForAll{(\varu, \vardlist) $\in$ \extendableIndex }{
            \ForAll{\varl $\in$ E'(\varu,\varv)}{
                \ForAll{(\varddu, \vardu) $\in$ \vardlist}{
                    \varddv = \varddu + $w1(\varl)$\\\label{ext1}
                    \vardv = \vardu + $w2(\varl)$\\\label{ext2}
                    \tcp{Filters: constraints and dist}
                    \If{\varddv $\leq$ c1 \kwand \vardv $\leq$ c2 \label{prefilter1}\\\kwand \vardv $<$ \vardistv[\varddv]\label{prefilter2}}{
                         \vardistv[\varddv] = \vardv\label{add1}\\
                         \If{\not \varpfcand[\varddv]}
                         {\varnb++\\
                         \varpfcand[\varddv] = \true\label{add2}\\
                         }

                         \If{\varddv > \varimax}
                         {\varimax = \varddv }
                    }
                }
            }
        }

    \Return \varnb, \varimax \\
\end{algorithm}

\begin{algorithm}
\footnotesize
    \caption{ \footnotesize CptExtendablePaths (\extendableIndexbisv,
     \varpfrontiv, \varpfcand, \var{d2\_it}, \vardistv)}\label{algo:ComputeExtendable}
        \varlastd = $\infty$\\
        \ForAll{\vardd $\in$ \var{d1\_it}}{
            \If{\vardistv[\vardd] $<$ \varlastd}{
                \add(\varpfrontiv, (\vardd,\vard))\\
                \varlastd = \vardistv[\vardd]\\
                \If{\varpfcand[\vardd]}{
                \add(\extendableIndexbisv, (\vardd,\vard))
                }
            }

        }
 \end{algorithm}

\textbf{The output of \bestcop.} When our algorithm terminates, the \var{pfront} array contains, for each segment number, all the distances of non dominated paths from the source $s$ towards each destination $d$. To answer the 2COP problem, for each $d$ and for all (stricter sub-)constraints $c_0'\leq c_0$, $c_1'\leq c_1$ and $c_2'\leq c_2$, we can proceed as follows in practice:
\begin{itemize}
    \item for $f(M_1, c_0', \Gamma, c_2', s, d)$, \ie to retrieve the distance from $s$ to $d$ that verifies constraints $c_0'$ and $c_2'$ minimizing $M_1$, we look for the \textit{first} element in
    $\varpfront[c_0' - 1][d]$ verifying constraint $c_2'$ (the first feasible distance is also the one minimizing $M_1$ because they are indexed on the later metric).
    \item for $f(M_2, c_0', c_1', \infty, s, d)$, we look for the \textit{last} element in
    $\varpfront[c_0' - 1][d]$ verifying constraint $c_1'$. The path minimizing $M_2$ being, by design, the last element.
    \item to compute $f(M_0, \infty, c_1', c_2', s, d)$, let us first denote $k$ the smallest integer such that $\varpfront[k][d]$ contains an element verifying constraints $c_1'$ and $c_2'$. The resulting image
     is then any of such elements in $\varpfront[k][d]$.
\end{itemize}
As one might notice, computing $f(M_j, c_0', c_1', c_2', s, d)$, $j=0,1,2$ cannot always be achieved in constant time (for $j=0$ and sub-constraints in particular). Indeed, we favor a simple data structure. A search in an ordered list of size $\Gamma$ is needed for stricter constraints (and may be performed $\log(MSD)$ times when optimizing $M_0$). To improve the time efficiency of our solution, each $\varpfront[i][d]$ may be defined as or converted into a static array in the implementation.

Finally, for simplicity, we did not show in our pseudo-code the structure and operations that store and extend the lists of segments. In practice, we store one representative of the best predecessors and a posteriori retrieve the lists using a backward induction for each destination.

\section{SR Graph \& Live Conversion Algorithm}\label{app:LCA}

The multi-metric SR Graph may be used in different ways. The path computation algorithm may be run directly on it. While this induces the cost of exploring a fully-meshed graph, it allows to treat the segment metric as a standard graph metric (in the form of the number of edges). Another way to use the SR Graph is to use the information it contains to perform conversions. The algorithm thus explores the original, sparser topology and converts the explored paths into segment lists on the fly to consider their number of segments. Note that the information within the SR Graph may be stored in different fashions (\eg by keeping the results of the APSP computation separated by source rather than merging it in a single SR Graph). However, the multi-metric APSP computation itself is mandatory.

The second method requires to design an efficient conversion algorithm, able to transform a path into a minimal segment list which respects the cost \emph{and} delay of the original path. This conversion is not trivial, as one must take into account the (forced) path diversity brought by ECMP. Indeed, using a node segment implies that the packet may follow any of the ECMP paths, which may possess heterogeneous delays. On the other hand, relying too much on adjacency segments will not lead to the minimal segment encoding of the considered paths. In addition, the number of segments, being an "off the graph" metric, exhibits peculiar behavior and requires to slightly extend the way path are checked for dominancy.

\subsection{Conversion Algorithm}
We start by introducing the following notations.
We consider a network $G = (V,E)$. 
Let $p = (x_0,\dots,x_l)$ be a path within $G$, with $x_i \in V$. Let $p_{[x_i, x_j]} = (x_i,\dots,x_j)$. We note $d(p)$ the couple of distances $(d_1(p), d_2(p))$ of the path $p$.


We denote $\mathit{Dag}$ a shortest path tree (regarding the $M_2$ metric, \ie the IGP cost). $\mathit{Dag}_{[u,v]}$ denotes the paths and distances from $u$ to $v$, within the shortest path tree $\mathit{Dag}$ rooted at $u$. To encode a given path, the $\mathit{Dag}$ rooted at each node $u \in V$ must be known. The latter must account for all ECMP path. Alternatively, one may compute the SR graph of the graph $G$, which contains the minimal amount of information required to perform any path encoding.

We denote $SR_G$ the SR graph of $G$, computed as explained in the main body of the paper (Section~\ref{sec:sr-graph}), 
but with the adjacency segments removed (the latter are not necessary to perform path encoding). We denote segments $S_i = (t_i,s_i, s_{i+1})$ with $s_i \in V$ and $t_i \in \{\mathit{AdjSeg}, \mathit{NodeSeg}\}$. The latter may either encode the best ECMP paths from $s_i$ to $s_{i+1}$ (in which case, $t_i = \mathit{NodeSeg}$) or a specific link between the two nodes (in which case, $t_i = \mathit{AdjSeg}$). We denote $SR(u,v)$ the best guaranteed distances when using a node segment from $u$ to reach $v$ (\ie the maximum delay among all path within $\mathit{Dag}_{[u,v]}$). In other words, it is the weight of the edge $(u,v)$ within $SR_G$. We define a list of segments as $S = (S_i)_{0 \leq i \leq l-1}$.

We now define the \emph{encoding} of a path $p$.
\begin{definition}[Encoding]~\label{def:encoding}
  Let $p = (x_0,\dots,x_l)$ be a path. Let $S = (S_i)_{0 \leq i \leq l-1}$ be a segment list. $S$ encodes the path $p$ if both conditions are verified:
  \begin{itemize}
    \itemsep0em
    \item $\forall i$, $\exists j$ such that $s_i = x_j$
    \item $\forall i$, either 
    \begin{itemize}
      \itemsep0em
      \item $p_{[s_i, s_{i+1}]}\in Dag_{[s_i, s_{i+1}]} \wedge SR(s_i, s_{i+1}) \leq d(p_{[s_i, s_{i+1}]})$
      \item $t_i = \mathit{AdjSeg}  \wedge (s_i, s_{i+1}) \in p$.
    \end{itemize}
  \end{itemize} 
\end{definition}



We then define the \emph{delay-cost path encoding problem}.

\begin{definition}[Delay-cost path encoding problem]
Given a path $p$, the delay-cost path encoding problem consists in finding a segment list $L$ that encodes $p$ with the minimal number of segments. 
\end{definition}

To solve the delay-cost path encoding problem, we start by presenting an algorithm which finds, given a path $p$, the longest prefix of $p$ encodable by a single segment. The pseudo-code of the algorithm is shown in Alg.~\ref{algo:1SLP}.

\begin{algorithm}
  \footnotesize
  \caption{ \footnotesize 1SegLongestPrefix($p$))}\label{algo:1SLP}
  \If{$p_{[x_0,x_1]} \notin \mathit{Dag}_{[x_0,x_1]}$}{\label{test1}
      \Return $(\mathit{AdjSeg}, x_0, x_1)$\\
  }

  \For{$i = 2\dots l$}{
    \If{$p_{[x_0, x_i]} \notin \mathit{Dag}_{[x_0, x_i]}$ \textbf{or} $SR(x_0, x_i) > d(p_{[x_0, x_i]})$}{\label{test2}
      \Return $\mathit{NodeSeg}, x_0, x_{i-1})$\\
    }
  }
  \Return $\mathit{NodeSeg}, x_0, x_l)$\\

\end{algorithm}



The test performed at Line~\ref{test1} checks whether the first 
edge of $p$ is within the shortest path tree of rooted at $x_0$. If so, the edge is not part of the ECMP path from $x_0$ to $x_1$. Thus, an adjacency segment is required. 
Otherwise, the algorithms checks, for each $i = 2\dots l$
whether the path $p_{[x_0, x_i]}$ is encodable in a single segment. This is equivalent to checking whether the path is within the shortest path tree rooted at $x_0$, and if the latter has, among all the ECMP paths $\mathit{Dag}_{[x_0, x_i]}$, the highest delay. Once this checks does not hold anymore, the algorithm returns.

Note that because the SR graph $SR_G$ encodes the minimal amount of information for such conversion, both conditions (Line~\ref{test1} and~\ref{test2}) can be checked easily and efficiently by relying on the SR graph $SR_G$. Indeed, the condition at Line~\ref{test2} can be expressed equivalently as $SR(x_0, x_{i-1}) + d(x_{i-1}, x_i) \neq SR(x_0, x_i)$. Similarly, Line~\ref{test1} can be expressed as $SR(x_0, x_1) \neq d(x_0, x_1)$.

We then generalize this algorithm to encode any given path $p$ with a minimal number of segments, as shown in Alg.~\ref{algo:encode}. 

\SetKwRepeat{Do}{do}{while}
\begin{algorithm}
  \footnotesize
  \caption{ \footnotesize encode($p$)}\label{algo:encode}
  $s := x_0$\\
  $S := \emptyset$\\
  \Do{$s \neq x_l$}{
    $S_i = (t, s, v)$ = 1SegLongestPrefix($p_{[s, x_l]}$)\\
    $L$.push($S_i$)\\
    $s := v$\\
  }
  \Return $L$
  
\end{algorithm}

Alg.~\ref{algo:encode} starts by finding the segment $S_i$ encoding the longest prefix in $p$ encodable in a single segment. Suppose that $S_i$ encodes the subpath $p_{[s, v]}$. The algorithm adds the segment to the segment list $S$ and repeats the procedure, considering the reminder of the path $p_{[v, x_l]}$ and $v$ as the new source. This is repeated until $v$ is equal to the last node in $p$, \ie $x_l$.\\

\begin{lemma} Let $S$ be a segment list that encodes $p$. Then, $|S| \geq |$encode($p$)$|$.
\end{lemma}

\begin{proof}
  Let us consider two segment lists $S = ((S_i)_{0 \leq i \leq l-1})$ and $S' = ((S'_i)_{0 \leq i \leq l'-1})$, both encoding $p$. Let $S' = encode(p)$ (\ie $S'$ is the result of Alg.~\ref{algo:encode}).

  For the sake of contradiction, let us assume that $l < l'$. Then, there must exist $k$, $0 \leq k \leq l-1$, such that 
  \begin{enumerate}
    \itemsep0em
    \item $s_k$ appears before, or at the same place as $s'_k$ in path $p$
    \item $s_{k+1}$ appears strictly after $s'_{k+1}$ in path $p$.
  \end{enumerate}

  Let us start by considering that $s_k = s'_k$. Then, 2. contradicts with the fact that $S'_k$ is a path encoding the longest prefix of $p$ (recall that $S_k'$ is the output of procedure 1SegLongestPrefix on $s_k'$).

  Otherwise, $s_k$ appears strictly before $s'_k$ in path $p$. Following the definition of encoding, $p_{[s_k, s_{k+1}]} \subset \mathit{Dag}_{[s_k, s_{k+1}]}$. In other words, $p_{[s_k, s_{k+1}]}$ is a shortest path (regarding $M_2$).
  Since $s'_k \in p_{\left[s_k, s_{k+1}\right]}$, then, $p_{\left[s_k', s_{k+1}\right]} $ is a shortest path as well, in other words
  \[
    p_{\left[s_k', s_{k+1}\right]} \subset \mathit{Dag}_{[s'_k, s_{k+1}]}
  \] 
  and by definition of $SR$,
  \[
    d(p_{[s'_k, s_{k+1}]}) = d(s'_k, s_{k+1}) \leq SR(s'_k, s_{k+1})
  \] 
 meaning that $(\mathit{NodeSeg}, s'_k, s_{k+1})$ encodes $p_{[s'_k, s_{k+1}]}$ which contradicts the fact that $S'_k$ is the longest encoding.

\end{proof}

Note that within our code, the conversion algorithm is implemented in a slightly different fashion for performance sake. During the exploration of the graph, path are extended edges by edges. To avoid recomputing the encoding of a path from scratch after each extension, we save for each path states allowing to efficiently update, for a given path $p$, the number of segments required if the latter is extended by an edge $e$. 


\subsection{Dominancy checks}

Although the algorithms presented in the previous section allow to find the number of segments to encode a path, it is not sufficient to correctly take into consideration the number of segments when computing paths. Indeed, conventional DCLC algorithm tend to only extend non-dominated paths. However, some precautions must be taken when comparing the number of segments of two paths. Indeed, because of the peculiar nature of the metric, the latter may evolve differently. As a consequence, paths that seem dominated must sometimes be extended nevertheless, as they may end up requiring less segments than their dominating competitors.

To illustrate this issue, let us consider two paths $p$ and $p'$. Let $S = (S_i)_{0 \leq i \leq l-1}= $ encode$(p)$ and $S' = (S'_i)_{0 \leq i \leq l'-1}= $ encode($p'$). 
Let $S'_{l-1} = (\mathit{AdjSeg}, s'_{l-1}, s'_l)$ and $S_{l-1} = (\mathit{NodeSeg}, s_{l-1}, s_l)$. Then, intuitively, if $p'$ is to be extended by an edge $e = (u,v)$, a new segment is mandatory, as $S'_{l-1}$ may not encode more than a single link. However, a new segment is not necessarily mandatory to extend $p$, as $S_{l-1}$ may very well be \textit{extended} to include $e$, if $p_{[s_{l-1}, v]} \in \mathit{Dag}_{[s_{l-1}, v]}$ and if $SR(s_{l-1},v) < d(p_{[s_l, v]})$.

Note that this effect may also occur (although less obviously) even when both $p$ and $p'$ end with a node segment. Indeed, it is possible that 
$p_{[s_{l-1}, v]} \in \mathit{Dag}_{[s_{l-1}, v]}$ while $p_{[s'_{l-1}, v]} \notin \mathit{Dag}_{[s'_{l-1}, v]}$

Consequently, even though a path may seem dominated, this may not be definitive. Thus, some paths that \textit{seem} dominated must be nevertheless be extended. More precisely, a path $p$ \emph{should not} be considered dominated by a path $p'$, regardless of their cost and delay, if the three following conditions are met:

\begin{enumerate}
  \itemsep0em
  \item $d_0(p) = d_0(p')$
  \item $s_{l-1} \neq s'_{l-1}$
  \item $t_{l-1} = \mathit{NodeSeg}$
\end{enumerate}

By modifying the dominancy condition as such, and relying on the conversion algorithms presented, a standard multi-metric path computation algorithm is able to correctly take into account the number of segments as another path metric. Although the overhead induced by our conversion algorithm should be limited, note that this dominancy condition implies that some paths that are ultimately dominated may have to be extended nonetheless,
increasing the overall worst-case complexity of the algorithm. The optimal usage of the SR Graph thus depends on the overhead of the conversion and extended dominancy check, compared to the cost of exploring a fully meshed graph.

\section{\bestcop upper bound with multi-threading}\label{upper_bound}
  
\begin{figure}
  \centering
  \includegraphics[scale=0.50]{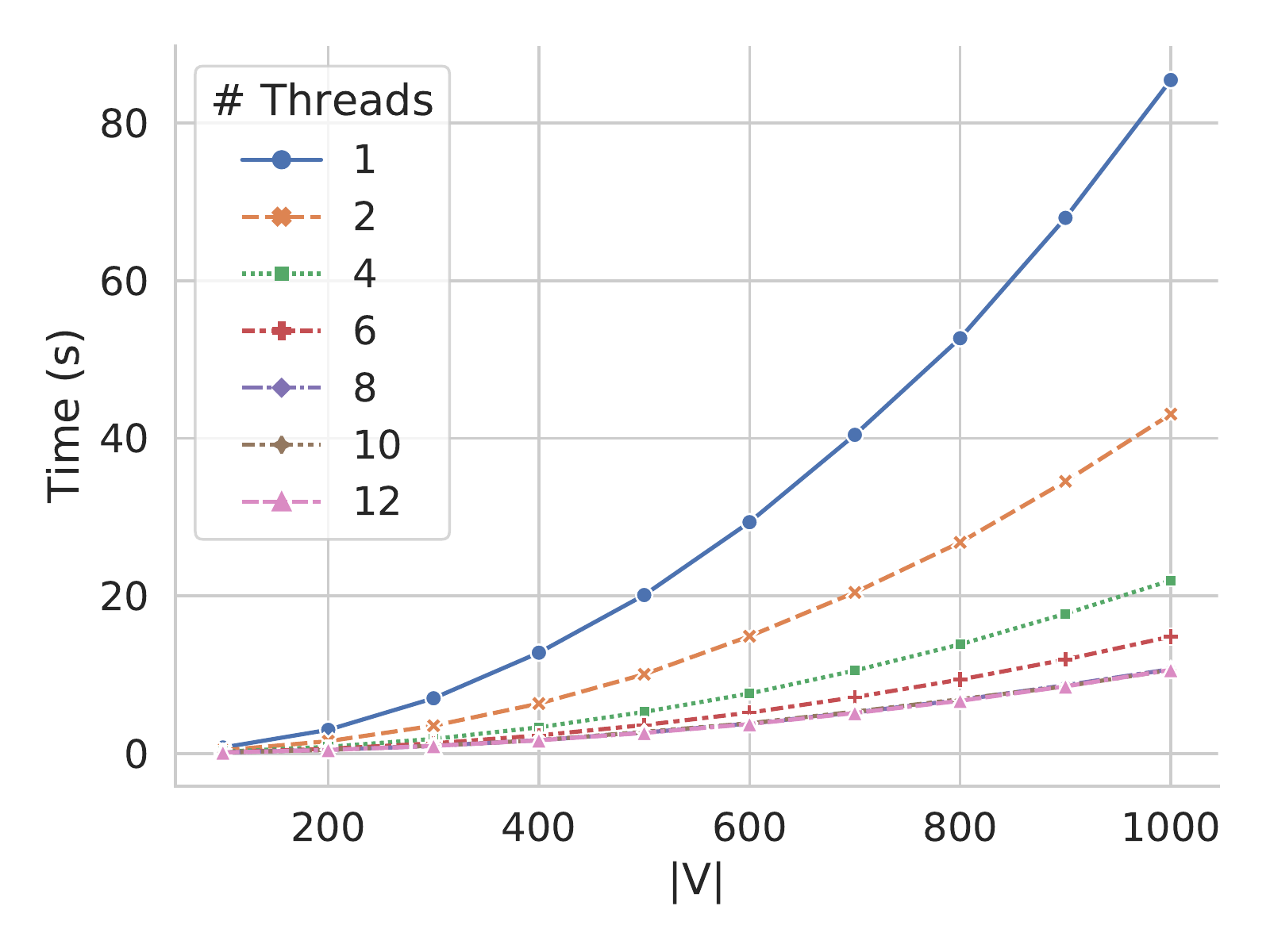}
  \caption{Upper bound of \bestcop execution time regarding the number of nodes and threads.}
  \label{fig:CT}
\end{figure}

To test \bestcop upper bound, we force it to explore its full iteration space. In other words, the algorithm behaves as if $|V|^2 \times L \times \Gamma$ distances have to be extended at each iteration. 
The results are shown in Fig.~\ref{fig:CT} for an increasing number of nodes and threads.
\bestcop does not exceed 84s when using a single thread, considering $|V| = 1000$ and $L = 2$, the average number of parallel links in the transformed graph.
This time, reasonable given the unrealistic nature of the experience, is significantly reduced when relying on multi-threading.
Using 8 threads, \bestcop execution time decreases to around 10s, highlighting both the parallelizable nature of our algorithm and its inherently good performance (past this number of threads, there is no more speedup as we exceed the number of physical cores of our machine). Additional experiments conducted on a high-performance grid show that \bestcop reaches a speedup factor of 23 when run on 30 cores.

\end{document}